%% file: arxiv.tex
\documentclass{article}

\usepackage[english]{babel}

\usepackage[letterpaper,top=2cm,bottom=2cm,left=3cm,right=3cm,marginparwidth=1.75cm]{geometry}

\usepackage[utf8]{inputenc}
\usepackage[T1]{fontenc}
\usepackage{microtype}

\usepackage{xcolor}
\definecolor{DarkBlue}{rgb}{0.1,0.1,0.5}

\usepackage{amsmath,amsthm,amsfonts}
\usepackage{graphicx}
\usepackage[colorlinks=true, allcolors=DarkBlue]{hyperref}
\usepackage{natbib}
\usepackage{tikz}
\usepackage{pgfplots}

\pgfplotsset{compat=1.10}
\usepgfplotslibrary{fillbetween}
\usetikzlibrary{patterns}
\usepackage{subcaption}
\usepackage{xcolor}
\usepackage[skip=2mm,indent=5mm]{parskip}

\usepackage{bbm}
\usepackage[ruled,vlined]{algorithm2e}

\usepackage{amsthm}

\newtheoremstyle{spaced}
  {6pt}       
  {6pt}        
  {\itshape}   
  {}           
  {\bfseries}  
  {.}          
  { }          
  {}           

\newtheoremstyle{spaceddefinition}
  {6pt}         
  {6pt}          
  {\normalfont}  
  {}             
  {\bfseries}    
  {.}            
  { }            
  {}             

\theoremstyle{spaced}
\newtheorem{theorem}{Theorem}[section]
\newtheorem{proposition}[theorem]{Proposition}

\newtheorem{lemma}[theorem]{Lemma}

\newtheorem{observation}[theorem]{Observation}
\newtheorem{remark}[theorem]{Remark}
\theoremstyle{spaceddefinition}
\newtheorem{definition}[theorem]{Definition}

\title{Stochastic wage suppression on gig platforms \\ and how to organize against it}

\usepackage{amssymb}

\usepackage{authblk}
\stepcounter{footnote}
\addtocounter{footnote}{-1}
\author[1]{Ana-Andreea Stoica}
\author[1,2]{Celestine Mendler-D\"unner}
\author[1]{Moritz Hardt}
\affil[1]{Max Planck Institute for Intelligent Systems, T\"ubingen, and T\"ubingen AI Center}
\affil[2]{ELLIS Institute Tübingen}
\date{}        
\begin{document}

\maketitle

\begin{abstract}
  Digital labor platforms are increasingly used to procure human input, ranging from annotating data and red-teaming AI models, to ride-sharing and food delivery. A central concern in such markets is the ability of platforms to suppress wages by exploiting the abundance of low-cost labor. To study this exploitation pattern, we introduce a novel posted-price procurement model with coverage objectives. A platform seeks to complete $M$  tasks by posting prices to sequentially arriving workers, each of whom accepts a task if it exceeds their private cost.
   First, we show that under natural assumptions on the workers' estimated cost, there exists a simple pricing strategy for the platform to cover all $M$ tasks with wait time $O(M)$, while paying only a $O(\log(M)/M)$ fraction of the total cost of labor.  This result highlights how platforms can exploit workers’ uncertainty about the cost of labor to effectively suppress wages.
  Then, we study collective action as a lever to increase wages and promote welfare in digital labor markets. In particular, we show how a small coalition of targeted low-cost workers who commit to a price floor forces the platform's total spending from logarithmic to linear in $M$. In contrast, a randomly sampled coalition of equal size remains largely ineffective. We complement our theory with synthetic experiments, showcasing the benefits of collective action across different market regimes.
\end{abstract}

\input{sec1-introduction}
\input{sec2-prelimsmodel}

\input{sec3-nocoordination}

\input{sec4-randomalphacoordination}
\input{sec5-targetedcoordination}

\input{sec6-conclusion}

\bibliographystyle{plainnat}
\bibliography{main}

\newpage
\appendix

\input{sec-appendix}

\end{document}

%% file: sec1-introduction.tex
\section{Introduction}
\label{sec:intro}

The rise of the gig economy has changed the way people get paid for work~\citep{taylor2017good}. Gig platforms increasingly fragment work into hundreds of fine-grained microtasks---from labeling data and verifying outputs of AI models to moderating online content, and delivering meals. At the same time, employers can post wages on platforms and dynamically adapt their pricing to worker supply.
As a consequence, workers are in a position where they constantly make decisions about whether to accept a task, or not, in a highly strategic and ever changing environment.

A growing concern is that platforms can exploit the high uncertainty of workers about task costs to suppress wages~\citep{hara2018data,muller2019algorithmic}. In fact, empirical research on gig platforms shows that workers not only exhibit high uncertainty around their own costs~\citep{hara2018data} but also often overlook the hidden costs of searching, monitoring, and payment management~\citep{crain2016invisible,toxtli2021quantifying}.~\citet{dubal2023algorithmic} argues that employers on crowdsourcing platforms are able to exploit this and ``[...] produce unpredictable, variable, and personalized hourly pay.'' However, while wage variability and underpayment have been documented empirically, the theoretical foundations linking pricing strategies to wage suppression remain underexplored.~\looseness=-1 

Collective action---through coordinated price floors, information campaigns, or shared organizing tools---has emerged as one of the few levers available to workers~\citep{irani2013turkopticon,woodcock2021fight,savage2020becoming,sigg2025decline}. Our work provides a formal model that explains when and how such coordination can counteract algorithmic wage suppression. We provide a theoretical framework studying wage suppression on gig platforms and design interventions that mitigate worker exploitation through coordinated action of low-cost workers.

\subsection{Our contributions}

\begin{figure*}
  \centering
  \subfloat[Market regimes]
  {\includegraphics[width=.45\linewidth]{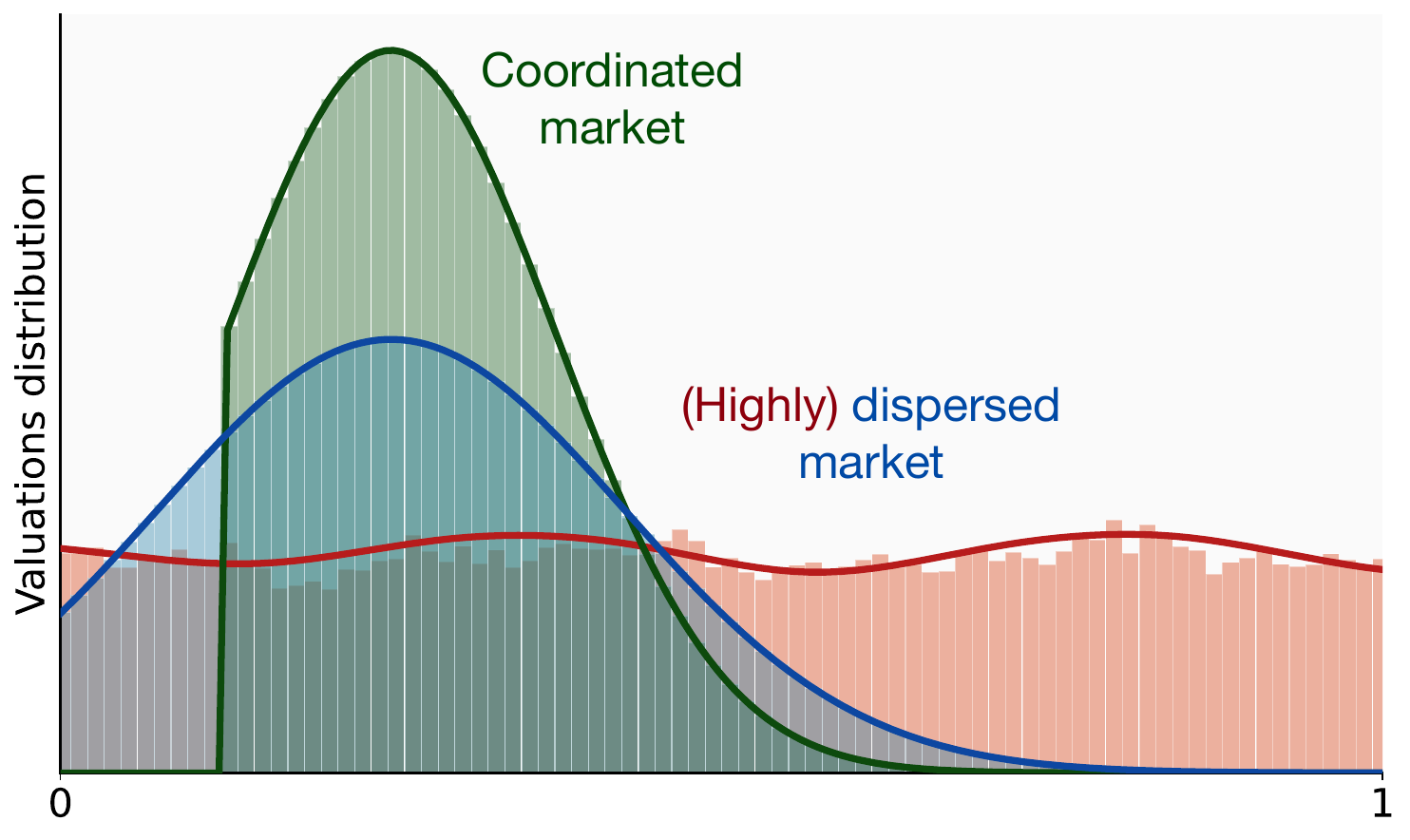}} 
  \hspace{0.05\textwidth}
  \subfloat[Total cost as a function of $M$]
  {\includegraphics[width=.45\linewidth]{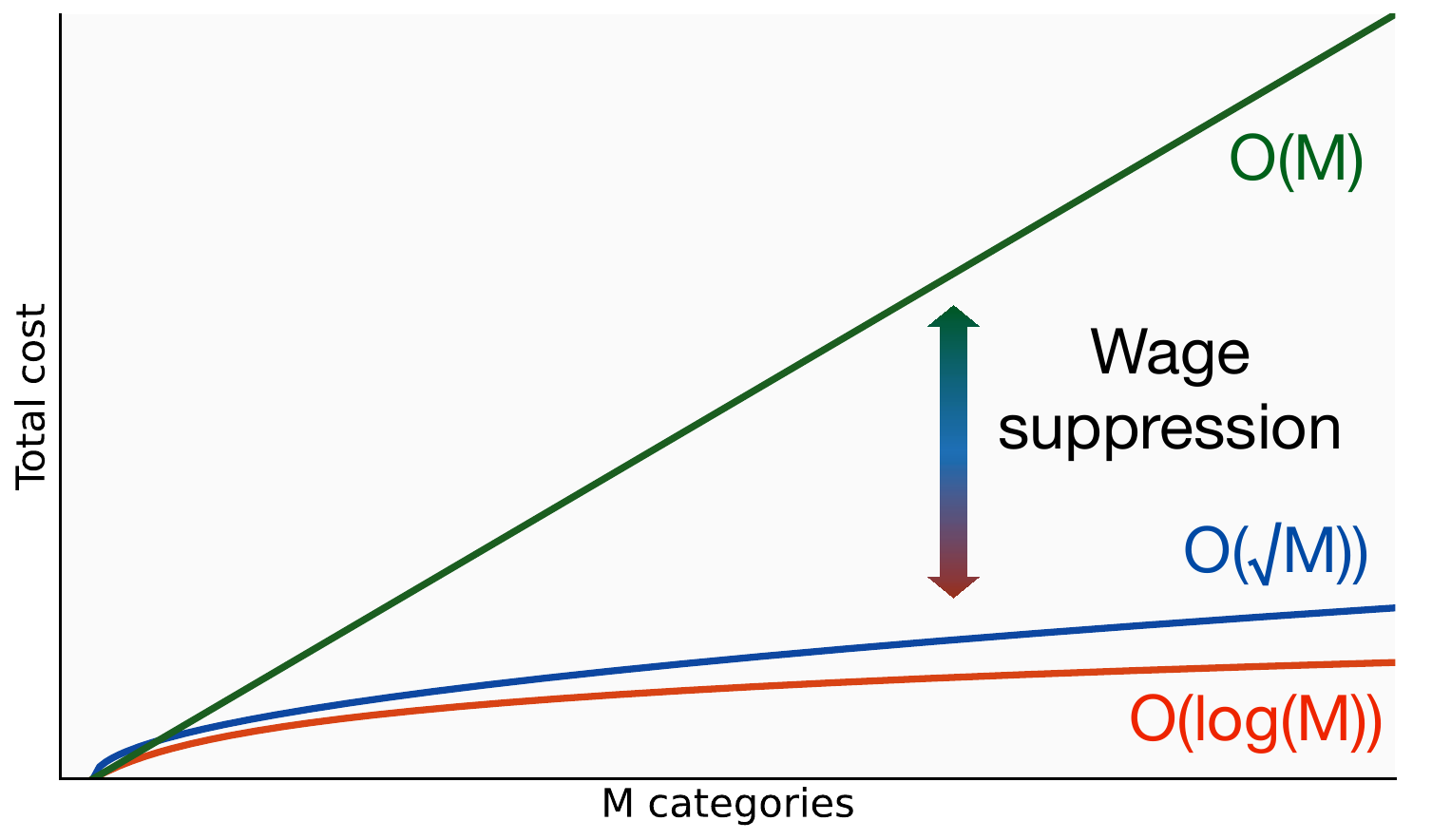}} 
  \caption{ An illustration of possible market regimes and corresponding wages under a linear wait time constraint. Dispersed markets (blue and red) are emblematic of workers' high uncertainty over their own costs and presence of workers willing to take arbitrarily low prices for tasks. In contrast, coordinated markets (green) represent cases where workers have coordinated on a price floor. A principal necessarily pays prices on the order of $O(M)$ in coordinated markets to achieve a linear wait time, however, he can exploit dispersed markets and reduces his total cost to sublinear, even logarithmic or constant with respect to $M$, preserving the same wait time guarantees.~\looseness=-1}
  \label{fig:illustration}
\end{figure*}

We propose a novel procurement model to understand wage suppression due to task uncertainty. In our model, a principal wishes to obtain the completion of tasks across $M$  categories. These categories might represent different types of data annotation (e.g., labeling images, annotating medical data) or tasks such as red-teaming exercises for machine-learning models. The principal represents a firm\footnote{Throughout our paper, we use the terms principal, firm, and platform interchangeably, leaving for future work the problem when \emph{multiple} firms compete on the same platform for attracting workers (e.g., multiple machine learning models competing for data annotators on MTurk or Prolific).} and repeatedly posts prices with the goal of attracting workers who will complete the tasks. Workers arrive sequentially and decide whether they are willing to complete the tasks for the given price. Upon completion the worker obtains the posted price, and the principal incurs the corresponding cost for payments together with the associated time they had to wait until a worker accepted the task.~\looseness=-1 

More formally, the process unfolds in iterations. At the start, all $M$ categories are active. At iteration~$t$, the principal posts a price $p_t$ that applies to any task from the currently active categories. Workers have a true cost associated with each task category, for which they have a noisy estimate, denoted as their \emph{valuation}. Workers draw their valuations from a set of distributions, one for each task category. These distributions model the workers' uncertainty in their cost estimate. Workers arrive sequentially and a worker accepts a task if and only if their valuation for one of the tasks is at most the posted price.~\looseness=-1

A \textbf{success event} for the principal occurs when a worker accepts a task. In this case, the principal pays the posted price $p_t$, and the corresponding category is removed from the active set. The  process continues with the reduced set of categories and a (possibly updated) price $p_{t+1}$. The \textbf{wait time} for iteration $t$ is the number of workers that had to be drawn until the first success.~\looseness=-1

We show that if there is a non-negligible price floor $p$ for each category (e.g., the market rate), the principal must pay a total cost of $pM$ to ensure coverage of all $M$ categories in finite time. However, if such a floor does not exist (i.e., workers accept a task at arbitrarily low wages), the principal can exploit the \emph{tails} of the valuations distributions and suppress wages while still preserving a wait time on the order of $O(M)$ with high probability. We denote this strategy as \emph{stochastic wage suppression} and analyze its implications for workers' total welfare:~\looseness=-1
\begin{itemize}
    \item We discuss properties of the workers' valuations distributions, connect these properties with empirical observations in the gig economy, and formally link them to different wage suppression regimes. Figure~\ref{fig:illustration} showcases our main results: workers who have a price floor obtain a total wage of $\Theta(M)$; as their uncertainty increases and low-wage workers appear, the principal can suppress wages down to a sublinear regime, while maintaining a wait time $O(M)$ with high probability.
    \item Next we study collective efforts of workers coordinating their costs. We show that such strategies are ineffective with only partial participation when workers are randomly recruited: if a random $\alpha$-fraction of workers coordinate on a price floor in each category, the wage suppression bounds remain unchanged for any fixed collective of relative size $\alpha<1$.  Borrowing from labor economics terminology, we denote this type of strategies as \emph{horizontal} strategies, as participants are recruited with equal probability across the entire population.
    \item In contrast, we show that \emph{vertical} strategies where we recruit a specific segment of the population are effective for a smaller number of participants: for any $\alpha$-fraction of categories, we find a closed-form condition for the minimum budget required to design a targeted intervention for the workers at the low-end of the valuations distributions. These targeted interventions restore wages from any sublinear regime to $\Theta(M)$, while maintaining a wait time of $O(M)$ with high probability.~\looseness=-1
\end{itemize}

Thus, building on a versatile model with minimal assumption, our work illustrates how the power of collective bargaining is dependent on the ability of a firm to wait out high-priced workers as well as the workers' ability to strategically coordinate. More specifically, we compare two types of collective action: \emph{horizontal} coordination in which workers in the collective are sampled randomly across categories versus \emph{vertical} coordination in which workers in the collective are sampled from a specific market segment. Whereas horizontal coordination essentially requires full worker participation to restore wages, vertical coordination is effective with only a fraction of workers joining the collective. We complement our theory with synthetic experiments.\footnote{Code is available online at \href{https://github.com/astoica/pricing_strategies_gig_economy}{https://github.com/astoica/pricing\_strategies\_gig\_economy}. All experiments use Monte Carlo simulations of various distributions averaged over $60$ runs and were run in Python 3.9 on local computers.} 

\subsection{Related work}
Our model is a dual to classical posted-price mechanisms~\citep{chawla2007algorithmic,hartline2009simple,chawla2010multi,babaioff2017posting,dutting2019posted,jullien2019information}, where some works assess monopolistic tendencies of platforms~\citep{eren2010monopoly,van2024distributionally}. In these models, workers exhibit a take-it-or-leave-it behavior based on their valuations and firms optimize to maximize average revenue under full or partial knowledge of valuations distributions. Recent works in this space optimize revenue regret bounds when the valuation distribution is unknown~\citep{leme2018contextual,cohen2020feature,liu2021optimal,tullii2024improved,bacchiocchi2025regret} or heterogeneous~\citep{cesa2019dynamic,lykouris2025contextual}, by developing dynamic pricing algorithms that learn this distribution (for a survey see~\cite{den2015dynamic}). Our work differs in two ways: first, it does not focus on learning the optimal pricing under no information about the valuation distribution (but rather, it analyzes the effect of allowing quantile information to a principal to minimize wages as a function of the market size); second, we analyze the effect of pricing strategies as the number of tasks increases, showcasing the effect on \emph{workers' wages} over time, as well as the effectiveness of collective action strategies in improving wages.~\looseness=-1 

In our model, instead of a seller maximizing revenue, a buyer (the principal, for example a platform or a model developer) minimizes procurement cost under coverage constraints by posting tasks with associated wages. Our model draws from economic models that study monopsony power in digital labor markets. The labor economics literature has documented that employers exercise significant wage-setting power due to search frictions and information asymmetries, especially in so-called ``thick markets,'' where there is a large number of buyers and sellers~\citep{manning2013monopsony,kingsley2015accounting,naidu2022labor,azar2024monopsony}, which motivates us to analyze wages as a function of the number of tasks set by a principal. Our posted-wage setting models truthfully gig economy platforms like Amazon Mechanical Turk (MTurk and Upwork~\citep{amazonmturk,upwork_milestones_2024}, which have been shown to exhibit substantial monopsony power. In particular,~\cite{dube2020monopsony} empirically estimates labor supply elasticities to quantify the wage markdown on MTurk, whereas~\cite{horton2010labor} directly estimate the workers' reservation prices distribution properties. In contrast, our model takes a complementary perspective by characterizing the pricing strategy a principal uses and their resulting total cost as a function of distributional properties of the workers' valuations.

Our model is closest in spirit to posted-price procurement models~\citep{badanidiyuru2012learning,singer2013pricing,singla2013truthful,charalampopoulos2025competitive}, but differs in objective and feedback: rather than maximizing the submodular value subject to a budget while learning an estimate of the optimal revenue in hindsight, we optimize the principal's cost to ensure coverage of all categories. An alternative framework models workers and tasks as a two-sided market with bilateral arrival uncertainty and creates optimal mechanisms for task assignment~\citep{dickerson2024matching}. 
We showcase that a simple pricing strategy based on distribution quantiles knowledge allows a principal to suppress prices with bounded guarantees. Quantile information has been recently used in pricing to show that it can lead to robust pricing menus~\citep{suzdaltsev2022distributionally,wang2025power}. Moreover, we analyze strategies for collective action that prevent a principal from exploiting low-wage workers, which differs from earlier works.~\looseness=-1

Our work is inspired by recent initiatives to organize workers on gig platforms, through information campaigns such as Turkopticon~\citep{irani2013turkopticon}, FairWork~\citep{graham2018towards}, ShiptCalculator~\citep{calacci2022bargaining}, platform interventions to set minimum wages on driving gig economy~\citep{parrott2018earnings,cachon2021pricing}, and self-organizing campaigns~\citep{doordash,salehi2015we,savage2020becoming,pricesurges,umney2024platform} (for a survey  see~\citet{sigg2025decline}). Recently, collective action has gained attention in the machine learning literature where researchers investigate new levers specific to algorithmic systems~\cite[c.f.,][]{vincent2021data, hardt2023algorithmic,baumann2024algorithmic, solanki2025crowding, gauthier2025statistical,zhu2025look}. 
Most theoretical works assume randomly sampled collectives, with recent ones focused on randomly drawing collectives from minority groups~\citep{ben2025fairness}. We refer to collectives formed by randomly recruiting workers as \emph{horizontal} collective action. In contrast, we refer to collectives in which workers are recruited from a specific market segment as \emph{vertical} collective action,  (in our model, in a particular task category or set of task categories). We compare these collective action strategies in our model. This is motivated by extensive literature on labor organizing~\cite{commons1955history,calmfors1988bargaining,olson2012logic} which compares economic outcomes under different collective formation principles. 
Our work complements these studies by demonstrating the benefits of targeted recruiting on digital platforms.~\looseness=-1

%% file: sec2-prelimsmodel.tex
\section{Market model}
\label{sec:prelim}

We model a simple market in which there is a set of workers $\mathcal{W}$. Each worker is associated with valuations over $M$ categories of data, denoting workers' perceived cost for completing a task in each category. We use $f_i$ to denote the pdf of workers' valuations distribution for category $i \in [M]$, with associated CDF $D_i$ and (left) quantile $Q_i(u):=D_i^{-1}(u):=\inf\{x:\,D_i(x)\ge u\}$ for $u\in[0,1]$. We denote $[M] = \{1,2,\cdots,M\}$. Hence, the workers’ valuations are i.i.d. draws from a product distribution $\mathcal{D} = D_1 \times \cdots \times D_M$ and bounded in $[0,1]$. We assume the workers behave rationally and accept a task that pays higher than their valuation (they have a positive profit), keeping the option to defer at zero profit.\footnote{Alternatively, a worker can choose the category with the largest margin $i^*=\mathrm{argmax}_{j \in A_t} p_t - v_{j}$, noting that our results hold since prices are posted \emph{ex-ante}.} As such, the valuations are equivalent to a \emph{reservation price}, a canonical quantity in models of labor markets that represents the minimum price for which a worker would accept a task~\citep{manski1999worker,burbano2016social,krueger2016contribution}.~\looseness=-1

The principal wishes to complete a set $A$ of $M$ tasks. To this end, they post a sequence of prices $\mathbf{p} = (p_1, \dots, p_M)$. For each price $p_t\in \mathbf{p}$, workers arrive sequentially and decide whether to complete one of the remaining tasks $a\in A$ or to defer. The price $p_t$ remains posted until the task is completed. Once a task is completed it is removed from the active set $A$. The number of deferred events at price $p_t$ is denoted as $N_t$. This leads to the  dynamic outlined in Algorithm~\ref{alg:model}. We denote by $q_t$ the probability of obtaining a completed task in iteration $t$, which depends on the pricing strategy and the workers' valuation distributions.~\looseness=-1

\begin{algorithm}[t]
\KwIn{Start with $A_1=[M]$. Define a sequence of prices $\mathbf{p}$:}
\nl For $t\geq1$ proceed as follows: \
\begin{enumerate}
    \item The principal posts a price $p_t$ for the completion of a task in the set of  active categories $A_t$. 
    \item Repeat until success: 
    \begin{enumerate}
    \item 
    A worker is drawn uniformly at random from a distribution of workers. The worker is characterized by a valuation profile $v$ with $v_j\sim D_j, \forall j$, each drawn independently.
    \item A ``success'' event is defined as 
    \[S = 1\big[\exists j \in A_t :  v_{j} \leq p_t\big].\] 
    Namely, there is a category $j$ for which the worker's valuation $v_j$---a proxy for his cost of the task---is lower than the principal's posted price $p_t$. 
    \begin{itemize}
        \item If $S=1$, 
        the worker completes a task for one of the available categories \[ i^* \in \{j \in A_t: v_j \leq p_t \}\] 
        \item If $S=0$ a new worker is sampled, independently, and the process repeats from (a).
    \end{itemize}
    \end{enumerate}
    \item The principal pays $p_t$ and the category $i^*$ gets removed from the set of categories that the principal is willing to pay for in the future: 
    \[A_{t + 1} = A_t \backslash \{i^*\}.\]
    We advance to iteration $t + 1$ as long as $A_{t+1}\neq\emptyset$. Let $N_t$ count the number of unsuccessful events that occurred at step $t$.
\end{enumerate}
\caption{Platform-worker interaction}
\label{alg:model}
\end{algorithm}

\medskip
\noindent \textbf{Wait time constraint.} A central quantity in our model is the wait time. 
The \emph{wait time} for the principal, for a sequence of prices $\mathbf{p}$ is 
\begin{equation}
{\rm WAIT}(\mathbf{p})= \sum_t N_t.
\end{equation}
Canonical works in matching markets and labor economics show that price and wait time for task uptake are often at a trade-off~\citep{mortensen1999new,mertikopoulos2020quick} and wait time varies as a function of market size~\citep{ashlagi2019matching}. To navigate this, empirical studies of the gig economy also show that employers often impose a wait time constraint through algorithmic means or penalties~\citep{maffie2024visible,johnston2024employment}.

In the following, we are interested in pricing strategies that guarantees the firm a linear expected wait time as a function of $M$, i.e.,~\looseness=-1

\begin{equation}\exists C>0: \mathbb{E}[{\rm WAIT}(\mathbf{p})]<CM\;\forall M,\label{eq:contraint}
\end{equation}
where randomness comes from the stochastic arrival process. Under this constraint a principal may care about the cost-time trade-off. We are particularly interested in this trade-off in the regime of large $M$ given the proliferation of microtasks in digital crowdwork platforms, whose number and diversity in categories continues to expand~\citep{alkhatib2017examining,wood2019good,gray2019ghost} as well as the rapid expansion of the gig economy into new market segments~\citep{garin2023evolution}. 

\medskip
\noindent \textbf{Payment and wage suppression.} The principal's \emph{total cost} for a pricing strategy $\mathbf{p}$ is the sum of the prices for all the success events (note that we are guaranteed a success event in each price stage by definition of the model, what varies is the wait time): 

\begin{equation}
    \mathrm{COST}(\mathbf{p}) = \sum\limits_{t = 1}^{M} p_t
\end{equation}

We assume that the platform does not take a cut, so the principal’s total cost equals the total wages earned by workers. 
Without much information about the workers' valuations, a principal is limited in their ability to minimize cost and therefore may use a noisy statistic such as pricing at the mean, median, or simply pay the maximum allowed price. Since valuations lie in $[0,1]$, paying a total cost of $M$ guarantees a wait time of exactly $M$. By contrast, access to information about valuation distributions allows a price-conscious principal to strategically `wait out' higher-priced workers and therefore pay below-market wages (a strategy observed in markets where employers have significant information about workers' reservation prices or labor supply~\citep{allon2012large,cachon2021pricing}).

We refer to regimes in which total payments scale sublinearly in $M$ as regimes of \emph{wage suppression}. We show that depending on the shape of the valuation distributions, even partial information---such as oracle access to quantile values---allows a principal to achieve sublinear total payments, including logarithmic or constant regimes in certain cases.~\looseness=-1

\subsection{Workers' task valuations}
\label{sec:char}

In our setup, the workers have valuations that correspond to their perceived cost for completing a task from a given category. The distributions $D_i$ for $i\in[M]$ model the workers' uncertainty over their own costs: for example, if a category $i$ corresponds to annotating images on Mechanical Turk, the workers will estimate their market rate of $p$ dollars per hour for this task category. 
In our model, the properties of these distributions characterize different market regimes and determine the principal's ability to suppress wages. We discuss some key characteristics of these distributions.

\medskip

\noindent \textbf{Dispersion.} In reality, in many market segments workers have uncertainty over their costs and thus exhibit high dispersion of perceived costs, a phenomenon observed by measuring workers' heterogeneity in reported reservation prices~\citep{chen2025reservation}. In online markets, unpaid work is shown to increase workers' uncertainty over their costs~\citep{gray2019ghost,saito2019turkscanner,toxtli2021quantifying} and lead to highly dispersed reported bids or accepted wages. We exemplify this case through the uniform distribution ${\rm Unif}[0,1]$. In contrast, a market with no uncertainty over $p$ would be equivalent to setting $D_i$ as the Dirac distribution at $p$ ($v \sim \delta_p$). Highly dispersed markets are vulnerable to wage suppression. A principal is able to exploit the positive mass of low-priced workers (who accept a task at arbitrarily low wages) and benefits from a pricing strategy that pays significantly under the mean or median. We provide conditions for distributions that exhibit dispersion to be exploited. In contrast, more concentrated distributions are more resilient to wage suppression. We illustrate dispersed markets in blue and red in Figure~\ref{fig:illustration}.

\medskip

\noindent \textbf{Right-skew.} A prevalent case in online markets features valuations heavily skewed toward zero, i.e., a significant mass of workers accept tasks at arbitrarily low wages. Empirically, reservation prices often exhibit high skewness~\citep{krueger2016contribution} (often exogenous to the market and influenced by previous wages and time in unemployment). In online markets specifically, wages are often skewed towards the low-earners~\citep{horton2010labor,hara2018data,hornuf2022hourly} with reported wage depreciation due to unpaid work. We note that such markets are particularly vulnerable to wage suppression by a principal who strategically sets low prices. We exemplify this case with the ${\rm Beta(a < 1,1})$ distribution and give a condition based on left-tail properties to show that a principal can achieve a constant total cost as a function of $M$ in such vulnerable market regimes.

\medskip

\noindent \textbf{Tail properties.} In our setting, the relevant behavior for characterizing wage suppression is the distributions' behavior near zero: we characterize the rate of growth of the total cost under strategic pricing as a function of tail growth rate near $0$ (i.e., a function characterizing the spread of workers who take arbitrarily low prices). An important special case is when there is a price floor $p^*$, i.e. a price below which no single worker will perform the task, and thus there is no left tail near $0$. When the market value is positive and accurately reflects workers' cost, this essentially imposes a \emph{floor reservation price}, seen in empirical studies of labor markets~\citep{krueger2016contribution} and collective action~\citep{jager2025collective}. We call these markets \emph{coordinated markets} and illustrate this case in green in Figure~\ref{fig:illustration}. As we will see the existence of a price floor makes markets resilient to wage suppression, limiting a principal's ability to strategically suppress prices while preserving his wait time guarantee. 

Next, we analyze a natural and versatile pricing strategy that guarantees linear expected wait time for any valuation distribution. For this strategy, we are interested in characterizing payments in different market regimes by analyzing the limit dependence of the principal's cost on $M$.

%% file: sec3-nocoordination.tex
\section{Stochastic wage suppression strategy}
\label{sec:nocoordination}

We propose a general pricing rule that ensures that the expected wait time is $O(M)$ with high probability for any valuation distribution. It is based on evaluating the left quantiles of the distribution, assuming that partial knowledge through oracle access is available to the principal: 

\begin{definition}[Stochastic wage suppression]
    We define the stochastic wage suppression pricing strategy $\textbf{p}^{\mathrm{SWS}}$ by setting the posted price at iteration $t$ to the $\theta/|A_t|$ quantile of the workers' valuations distributions:
    \begin{equation}
    p_t \;=\; \max_{i\in A_t} Q_i\!\Big(\frac{\theta}{\,|A_t|\,}\Big) 
    \qquad t=1,2,\dots,M,
\end{equation}
for some $\theta \in (0,1]$. In our model, exactly one category is removed at each iteration, so $|A_t| = M - t  +1 $ by construction.
\label{def:SWS}
\end{definition}

We denote this strategy as the \emph{stochastic wage suppression pricing strategy} because it exploits the tails of the valuation distribution to reduce payments for the principal.

\subsection{Linear wait time guarantee }

Our first main result is that the $\textbf{p}^{\rm SWS}$ satisfies the linear expected wait time constraint~\eqref{eq:contraint} for any valuation distribution. To show that, we note a sufficient condition that lower-bounds the probability of success for each iteration: 

\begin{lemma}[Lower bound on success probability]
    \label{lemma:quantileprices}
    For the stochastic wage suppression pricing strategy $\textbf{p}^{\mathrm{SWS}}$, there exists a constant $q^*>0$ independent of $M$ for which the success probabilities satisfy $q_t \geq q^*, \forall t \in [M].$ 
\end{lemma}

\begin{proof}[Proof]
We can compute the success probabilities as 
\begin{equation}
\begin{gathered}
q_t \;=\; 1-\prod_{i\in A_t}\bigl(1-D_i(p_t)\bigr)
\;\ge\; 1-\Bigl(1-\tfrac{\theta}{M-t+1}\Bigr)^{M-t+1} \\
\;\ge\; 1-e^{-\theta}
\;=:\; q_{*}, \forall t \in [M],
\end{gathered}
\end{equation}
since for each active $i$, $D_i\bigl(p_t\bigr)\ge D_i\!\bigl(Q_i(\theta/|A_t|)\bigr)=\theta/|A_t|$. By construction, $|A_t| = M - t + 1, \forall t \in [M]$. Independence across categories gives
$q_t=1-\prod_{i\in A_t}(1-D_i(p_t))\ge 1-(1-\theta/|A_t|)^{|A_t|}\ge 1-e^{-\theta}$. Here we used that $(1 - \theta/x)^x \leq e^{-\theta}, \forall x,\theta > 0$. The $N_t$ are Geometric random variables since the process is sequential and each worker is drawn independently, $N_t\sim{\rm Geom}(q_t)$ on $\{1,2,\dots\}$. 
\end{proof}

\begin{theorem}[Wait time]
Consider the stochastic wage suppression strategy $\mathbf{p}^{\rm SWS}$ defined in Definition~\ref{def:SWS}, and let $q_{*}=1-e^{-\theta}$.
Then
\begin{equation}
\mathbb{E}[{\rm WAIT}(\mathbf{p}^{\rm SWS})]\;\leq\; \frac{M}{q_{*}} \;=\; O(M).
\end{equation}
Moreover, the iteration-wise wait times $\{N_t\}_{t=1}^M$ are independent sub-exponential random variables with parameters depending only on $q_{*}$.
Hence, for any $\delta\in(0,1)$,
\begin{equation}
{\rm WAIT}(\mathbf{p}^{\rm SWS}) \;=\; O\!\Big(M+\sqrt{M\log(1/\delta)}+\log(1/\delta)\Big)
\end{equation}
with probability at least $1-\delta$.
\label{thm:quantimeconcentration}
\end{theorem}

\begin{proof}[Proof sketch]
We use a standard result on Geometric random variables that are sub-exponential with $\psi_1$-norm bounded by a constant depending only on $q_{*}$ (shown in Lemma~\ref{lem:geom-subexp-stage} in the Appendix). We then use the independence of the wait times at every iteration to compute the total wait time in expectation as $\mathbb{E}[{\rm WAIT}(\textbf{p}^{\rm SWS})]\le M/q_{*}$. Finally, we show that the total wait time ${\rm WAIT}(\textbf{p}^{\rm SWS})=\sum_t N_t$ satisfies a Bernstein (sub-exponential) tail and we use Bernstein's inequality to prove the concentration bound.~\looseness=-1
\end{proof}

The $\textbf{p}^{\rm SWS}$ strategy has a natural interpretation through quantile access at the inverse number of active categories, which gives the constant lower bound of success in each iteration. When $M = 2$, $\textbf{p}^{\rm SWS}$ is pricing at the median for $\theta = 1$; as $M$ increases, access to lower quantile values allows the principal to price under the median while maintaining the wait time.
Figure~\ref{fig:totaltime_quantile_all} shows the total wait time under the $\textbf{p}^{\rm SWS}$ strategy, noting that the $O(M)$ time guarantee holds, with different constants depending on the exact valuations distributions.~\looseness=-1

\begin{figure}
  \centering
  \includegraphics[width=0.45\linewidth]{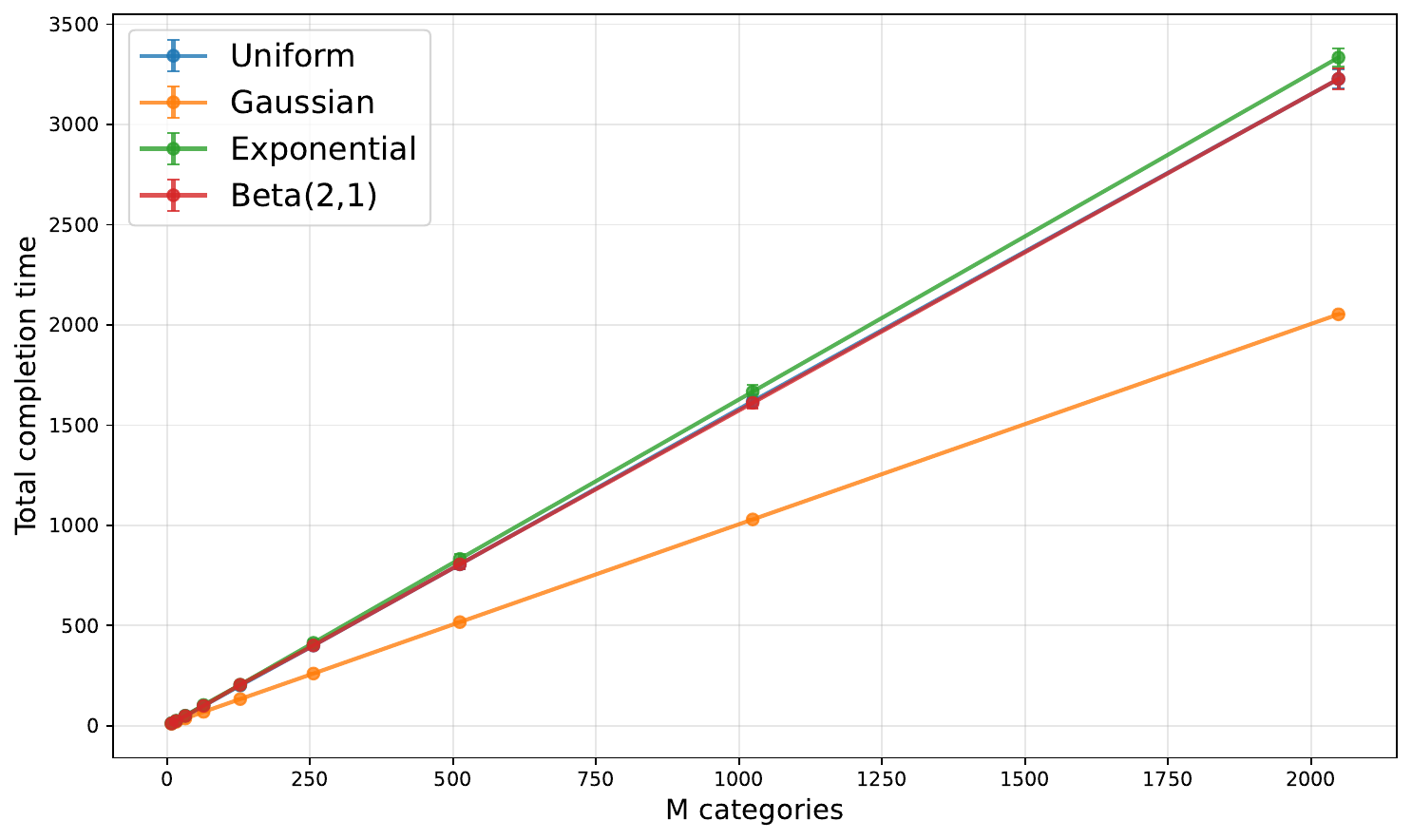}
  \caption{Total wait time for the $\mathbf{p}^{\rm SWS}$ pricing strategy for different distributions, as $M$ varies.} 
  \label{fig:totaltime_quantile_all}
\end{figure}

\subsection{Total payments}
While holding the expected wait time fixed, the stochastic wage suppression pricing strategy achieves a different total cost depending on the distribution of valuations among workers. 
In the following we use the characteristics outlined in Section~\ref{sec:char} to describe different cost regimes.~\looseness=-1

\medskip

\noindent \textbf{No wage suppression.} In coordinated markets where all workers have a minimum positive market price for each category $i \in [M]$ (e.g., there is a minimum wage requirement or there is no uncertainty or dispersion in the workers' costs), the principal pays a total price of $\Theta(M)$. We characterize this property as a sufficient condition on the total cost bound below: 

\begin{lemma}[Linear total cost]
\label{lemma:gap-linear}
For each category $i$, let $g_i := \inf\{x\in[0,1]: D_i(x)>0\}$ denote the left endpoint of the support (informally, the gap at zero price). There is a positive left gap if 
\begin{equation}
\label{eq:left-gap}
g_{*} \;:=\; \min_{i\in[M]} g_i \;>\; 0.
\end{equation}
For any distributions $\{D_i \}_{i\in [M]}$ with property~\eqref{eq:left-gap}, the $\textbf{p}^{\rm SWS}$ pricing strategy pays a total cost of $\Theta(M)$.
\end{lemma}

Examples of such distributions include point-mass distributions at $p_j > 0$ for each category $j$ (e.g., a collective contract that sets wages in each market segment) or truncated distributions for which there is no worker density below a price floor $p >0$. Such markets can occur naturally when the workers do not have uncertainty over their costs. 

\medskip 

\noindent \textbf{Sublinear wage suppression regime.} Empirical research shows that markets in the gig economy are not often coordinated. Without the positive left gap property, we transition to a regime where the total cost will achieve a sublinear bound:

\begin{lemma}[Sublinear total cost]
If for each $i \in [M]$,
\begin{equation}
\inf\{x\in[0,1] : D_i(x) > 0\} = 0,
\label{prop:gapless}
\end{equation}
the $\mathbf{p}^{\rm SWS}$ pricing strategy pays a total cost of $o(M)$.~\looseness=-1
\label{lemma:gapless}
\end{lemma}

A large array of distributions fall under this regime. For example, CDFs with polynomially vanishing density at $0$ such as a power left tail: $\exists a > 1, K > 0$ s.t. $D(x) = K x^a(1 + o(1))$ as $x \rightarrow 0$. A specific example is the ${\rm Beta}(a>1,1)$ distribution, noting that varying $a$ achieves a sublinear bound $\Theta(M^{1 - 1/a})$ with any exponent (see Observation~\ref{obs:sublinear} in the Appendix for computation and further examples). As such, a positive density at $0$ implies a phase transition in the total wages to become sublinear.~\looseness=-1

\medskip

\noindent \textbf{Logarithmic wage suppression regime.} A further transition occurs when the left tail grows linearly, as a special case of the aforementioned polynomial; in this regime, the stochastic wage suppression strategy achieves a total cost bound of $O(\log(M))$. As described in Section~\ref{sec:prelim}, this regime can be exemplified as \emph{high dispersion markets} where workers have high uncertainty over their costs.
We state a distribution-agnostic sufficient regularity condition that allows a strategic principal to achieve a cost of $O(\log M)$ while preserving the linear wait time guarantee. Informally, this condition implies a ``linear left-tail growth,'' i.e., $D_i(x)\gtrsim x$ near $0$.~\looseness=-1 

\begin{lemma}[Logarithmic total cost]
    We denote by left-tail regularity the following condition: there exist constants $c>0$ and $x_0\in(0,1]$ such that for all $i\in[M]$ and all $x\in[0,x_0]$, 
    \begin{equation}
    D_i(x)\;\ge\; c\,x.
    \label{condition:left}
    \end{equation}
    Equivalently, for all $u\in[0,cx_0]$, $Q_i(u)\le u/c$, or $\mathbb{P}(v_i < 1/t) = \Omega(1/t), \forall t \in [M]$.
    For any distributions $\{ D_i \}_{i \in [M]}$ with property~\ref{condition:left}, the $\textbf{p}^{\rm SWS}$ pricing strategy pays a total cost of $O(\log(M))$. 
    \label{lemma:logm}
\end{lemma}

Examples of distributions that fall under property~\ref{condition:left} are the uniform distribution, the exponential distribution, and any left-skewed power law with a heavy tail, such as ${\rm Beta}(1, b)$ for $b < 1$ (see Observation~\ref{obs:disperseddistributions} in the Appendix for exact computations). 


\noindent \textbf{Constant wage suppression regime.} A further phase transition occurs where there is a significant mass concentrated near $0$ (so, most workers will have valuations arbitrarily close to $0$): in this regime, the principal's pricing strategy suppresses wages down to a total cost of $O(1)$ with respect to $M$. This regime is emblematic of distributions exhibiting a \emph{right-skew} described in Section~\ref{sec:prelim}, which characterizes \emph{vulnerable markets}. We express a formal criterion below, with a generalization in the Appendix (Proposition~\ref{prop:O1-iff}).~\looseness=-1

\begin{lemma}[Constant total cost]
    If distributions $D_i$ have a heavier than linear left tail, such as a polynomial, $\exists x_0 \in (0,1], b \in (0,1), c > 0$ s.t. $\forall i\in [M], x \in [0, x_0]$ 
    \begin{equation}
    D_i(x)\;\ge\; c\,x^b.
    \label{condition:poly}
    \end{equation}
    Equivalently, for all $u\in[0,cx_0]$, $Q_i(u)\le (u/c)^{1/b}$, or $\mathbb{P}(v_i < 1/t) = \Omega(t^{-b}), \forall t \in [M]$.
    For any distributions $\{ D_i \}_{i \in [M]}$ with property~\ref{condition:poly}, the $\textbf{p}^{\rm SWS}$ strategy pays a total cost of $O(1)$ with respect to $M$. 
    \label{lemma:constant}
\end{lemma}

In addition, any distribution with an atom at $0$ will also satisfy the same properties from Lemma~\ref{lemma:constant}.~\looseness=-1 

We remark on the phase transition that the principal's strategy achieves: without collective action, pricing at the quantiles obtains sublinear total cost for the principal; depending on the exact distribution family, the total cost can vary from any sublinear bound, down to a logarithmic, and even a constant wage suppression regime as a function of $M$. These regimes mirror reported empirical properties of workers' valuations distributions in the gig economy. We provide sufficient conditions to characterize the principal's ability to suppress wages in different regimes. Figure~\ref{fig:totalcompletiontime_all_quantile} presents the total cost that $\textbf{p}^{\rm SWS}$ achieves for different regimes (top row, in blue for $\alpha = 0$, equivalent to no collective action). Different wage suppression regimes are exemplified through different distributions: $O(1)$ in panel (a), $O(\log(M))$ in panels (b)-(c) and $O(M^{1/2})$ in panel (d) for the total cost.~\looseness=-1

\begin{remark}[A note on cost optimality]
    The $\textbf{p}^{\rm SWS}$ strategy is cost-optimal among all strategies satisfying the per-iteration success probability constraint $q_t \geq q^*$ when distributions are identical across categories. This follows directly: with $|A_t| = M - t + 1$ categories at iteration $t$, the minimum price ensuring $q_t \geq q^*$ is $p_t = Q\left(1 - (1- q^*)^{1/|A_t|} \right)$, which is exactly the price defined by the $\textbf{p}^{\rm SWS}$ strategy. When these distributions come from a heterogeneous family of distributions, the principal may further optimize his cost by exploiting the heterogeneity (e.g., via contextual bandits~\citep{lykouris2025contextual}). Methods for optimizing the absolute cost of the principal are left for future work. 
\end{remark}

%% file: sec4-randomalphacoordination.tex
\section{Horizontal collective action fails to prevent wage suppression}
\label{sec:randomcollact}

In this section, we explore \textbf{untargeted} collective action strategies: in these strategies, a random fraction of the workers coordinate on a price floor (similar to the process described in~\citet{hardt2023algorithmic}), essentially only accepting tasks that are offered by the principal above a certain price. If all workers fully coordinate on all categories, that imposes a minimum price floor and thus prevents wage suppression as the total cost is $\Theta(M)$. Here, we are interested in the power of the collective as a function of its size. 

\paragraph{Horizontal $\alpha$-fraction coordination.}
Let $\alpha\in[0,1)$ denote the participation rate in the collective. A collective sets a category-specific minimum price $\overline{p}_j, \forall j \in [M]$. Since each arriving worker independently belongs to the collective with probability $\alpha$,\footnote{Note: this is consistent with sampling with replacement from a population with collective fraction $\alpha$.} we model their valuations as drawn from the mixture distribution:
\begin{equation}
v_j \;\sim\; (1-\alpha)\cdot D_j \;+\; \alpha \cdot \delta_{\overline{p}_j},
\end{equation}
where $\delta_{\overline{p}_j}$ is a point mass at the collective-set price $\overline{p}_j\in[0,1]$. Then, the CDF of this mixture distribution is
\begin{equation}
  D_j^{\mathrm{mix}}(x) \;=\; (1-\alpha)\cdot D_j(x) \;+\; \alpha\cdot \mathbf{1}\{x\ge \overline{p}_j\},\qquad x\in[0,1].
\end{equation}

We show that under this data-generating distribution, the stochastic wage suppression pricing strategy achieves a wait time for the principal of $O(M)$ with high probability, no matter the size of the collective (the value of $\alpha$). Essentially, the random collectives \emph{thin} the arrival rate of successful workers, but leave a positive mass of workers who take up a task at low prices which can continue to be exploited by a principal with only a constant factor in the wait time.~\looseness=-1

\begin{lemma}
Let $\mathbf{p}_{\alpha}^{\mathrm{SWS}}$ denote the SWS pricing strategy applied to the mixture distributions $\{D_j^{\mathrm{mix}}\}_{j\in[M]}$.
Then
\begin{equation}
\begin{gathered}
\mathbb{E}[{\rm WAIT}(\mathbf{p}_{\alpha}^{\rm SWS})]\;\le\;\frac{M}{q_{*}(\alpha)} \;=\; O(M), \\
{\rm WAIT}(\mathbf{p}_{\alpha}^{\rm SWS}) \;=\; O\!\Big(M+\sqrt{M\log(1/\delta)}+\log(1/\delta)\Big)
\end{gathered}
\end{equation}
with probability at least $1-\delta$, where $q_{*}(\alpha) := 1-\exp\!\bigl(- (1-\alpha)\,\theta\bigr)$ and with constants depending only on $\theta$ and $\alpha$.
\label{lemma:horizontaltime}
\end{lemma}

\begin{proof}[Proof sketch]
We note that the success probability in each iteration will still be lower bounded by a constant, which now depends on $\alpha$. This gives the same computation for the expected wait time, scaled by an $\alpha$-dependent constant, and the concentration inequality follows exactly like in Theorem~\ref{thm:quantimeconcentration}.
\end{proof}

Furthermore, the effect of the random collective only affects the principal's strategy by a constant in the total cost for any $\alpha < 1$:~\looseness=-1

\begin{lemma}
Fix $\alpha \in [0,1)$ and $\theta>0$, and let $\overline p^{-} := \min_{i\in[M]} \overline p_i$ denote the minimum collective price floor (assumed positive).
Then, across all market regimes analyzed in Section~\ref{sec:nocoordination},
\begin{equation}
{\rm COST}(\mathbf{p}_{\alpha}^{\rm SWS}) \;=\; \Theta\!\bigl({\rm COST}(\mathbf{p}^{\rm SWS})\bigr)
\end{equation}
as a function of the number of categories $M$.
\label{lemma:randomcollectivecost}
\end{lemma}

\begin{proof}[Proof sketch]
    The mixture distribution forces the stochastic wage suppression strategy to increase prices, as the quantile values shift: the new quantile values are either the original quantile values of the original distributions $D_i$ shifted by a factor of $1/(1-\alpha)$ (when the worker is not part of the collective), or equal to the minimal price floor $p^{-}$. In total, this shifts the total cost by a factor of $1/(1-\alpha)$ that follows from quantile reweighting under the mixture distribution and a constant addition for the collective workers. For the regimes described in Section~\ref{sec:nocoordination}, this translates to ${\rm COST}(\mathbf{p}_{\alpha}^{\rm SWS}) = \frac{1}{1-\alpha}{\rm COST}(\mathbf{p}^{\rm SWS})+ O(1)$ (full proof in the Appendix). Thus, the total cost will undergo the same phase transitions as described in the case without any collective action.~\looseness=-1
\end{proof}

\begin{figure}
  \centering
  \includegraphics[width=0.45\linewidth]{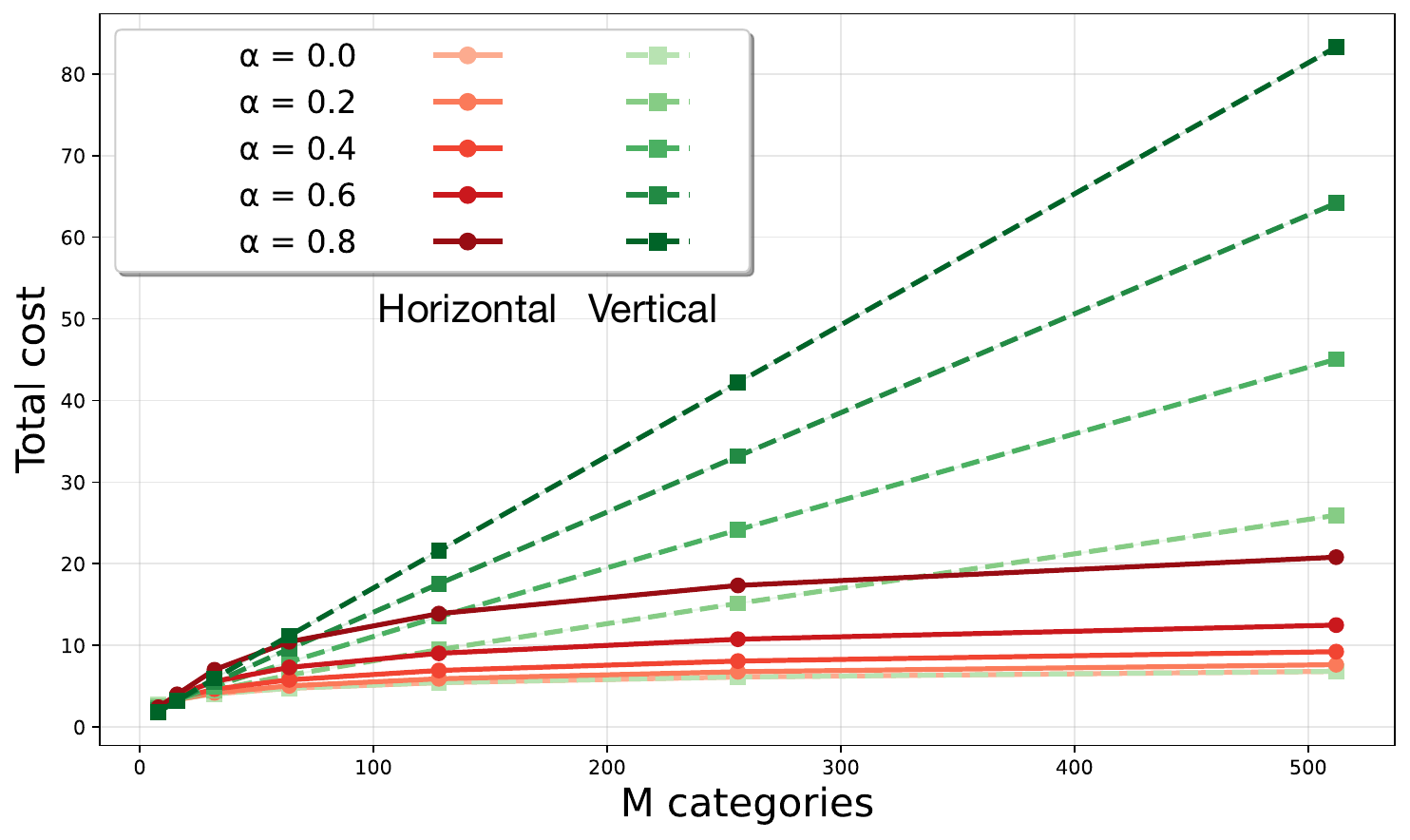}
  \caption{A comparison of the total cost achieved by a horizontal collective vs a vertical collective with a budget of $\epsilon = 0.01$ under the $\mathbf{p}^{\rm SWS}$ pricing strategy and uniform valuations distributions.}
  \label{fig:horizontalverticalcollact}
\end{figure}

Figure~\ref{fig:totalcompletiontime_all_quantile} showcases the total cost (top row) and the total wait time (bottom row) for various values of $\alpha$ for the $\mathbf{p}_{\alpha}^{\rm SWS}$ pricing strategy. The collective size does not affect the asymptotic growth rate of cost or wait time, but it does change constants (computed in closed form in the Appendix): larger collectives increase both cost and wait time by constant factors.~\looseness=-1

\begin{remark}[Information campaigns]
    The horizontal collective action model is inspired by decentralized organizing mechanisms~\citep{hardt2023algorithmic,umney2024platform}. Our results show that such uncoordinated efforts cannot overcome structural wage suppression: 
    recruiting anything less than all low-valuation workers preserves our bounds qualitatively, since horizontal coordination primarily \emph{thins} the arrival rate of low-cost acceptances. For this reason, while information campaigns have been deemed helpful for worker organizing~\citep{irani2013turkopticon,woodcock2021fight}, their efficacy may be limited if they fail to reach essentially all lowest-paid workers. In our model, information about a category's market price may reduce variance yet still leave tail properties that sustain sublinear or constant wage suppression regimes.~\looseness=-1
\end{remark}

\begin{figure*}
  \centering
  \subfloat[Beta$(0.5,1)$]
  {\includegraphics[width=.25\linewidth]{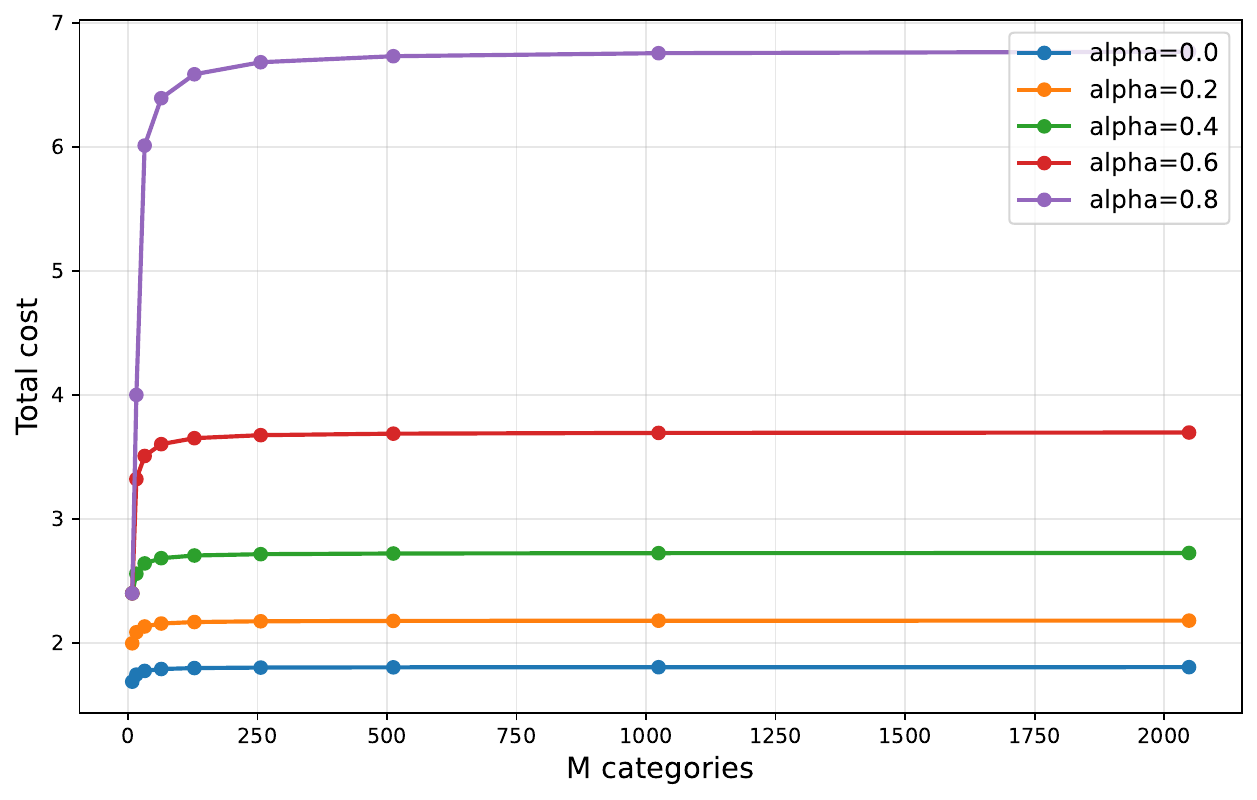}}
  \subfloat[Uniform]
  {\includegraphics[width=.25\linewidth]{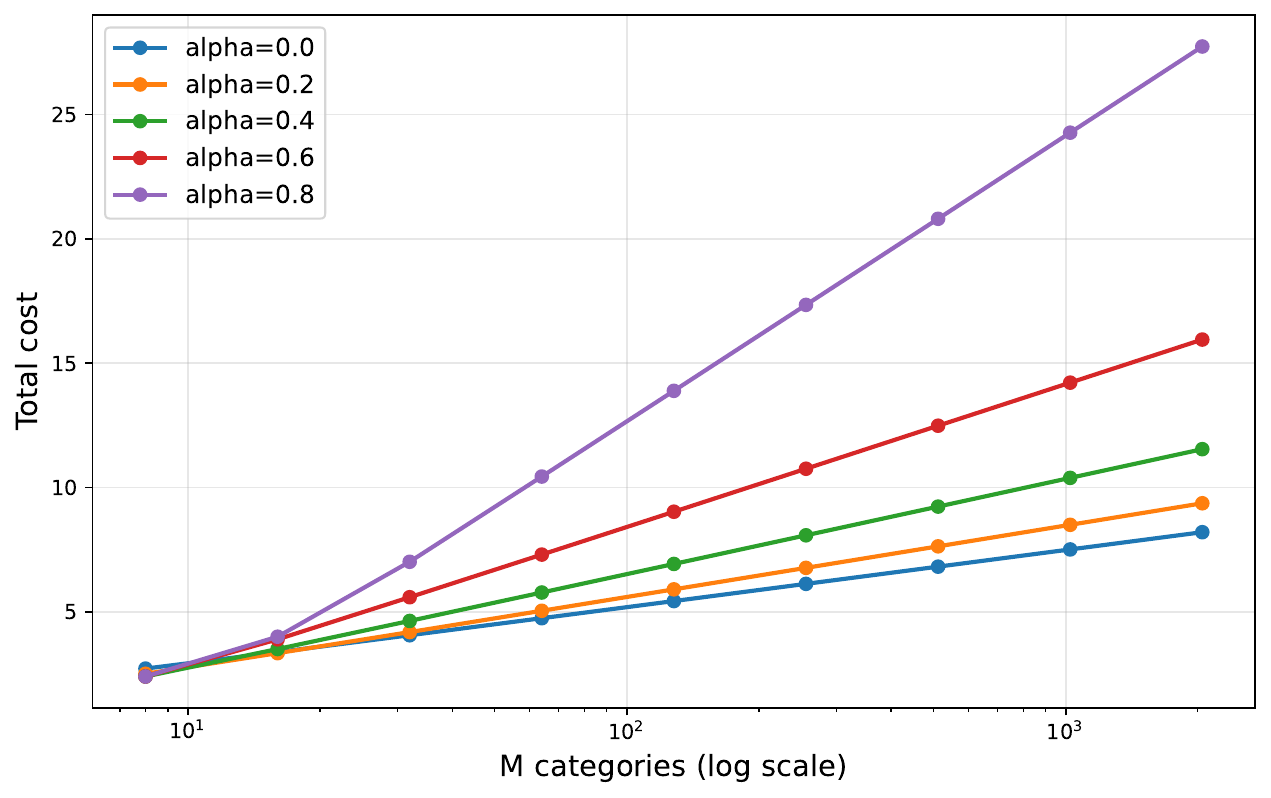}}
  \subfloat[Exponential, $\lambda = 3$]
  {\includegraphics[width=.25\linewidth]{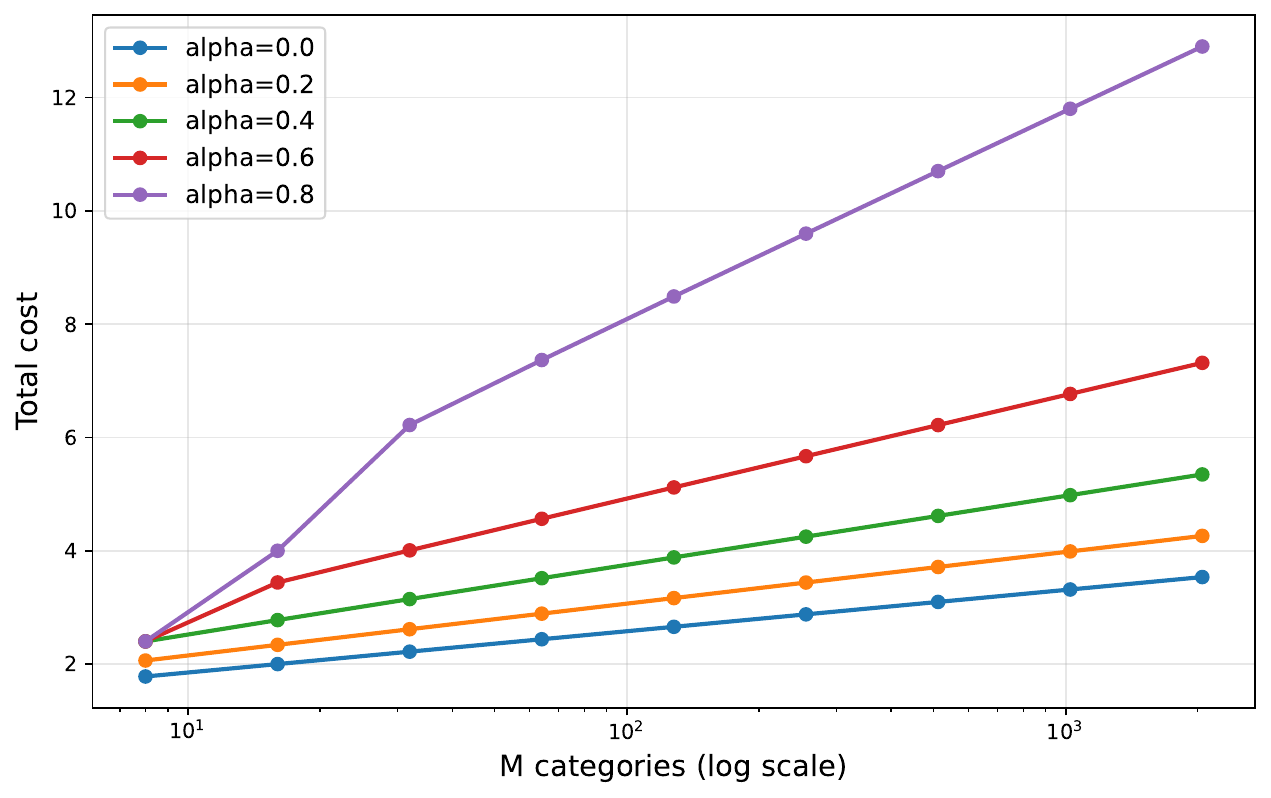}}
  \subfloat[Beta$(2,1)$]
  {\includegraphics[width=.25\linewidth]{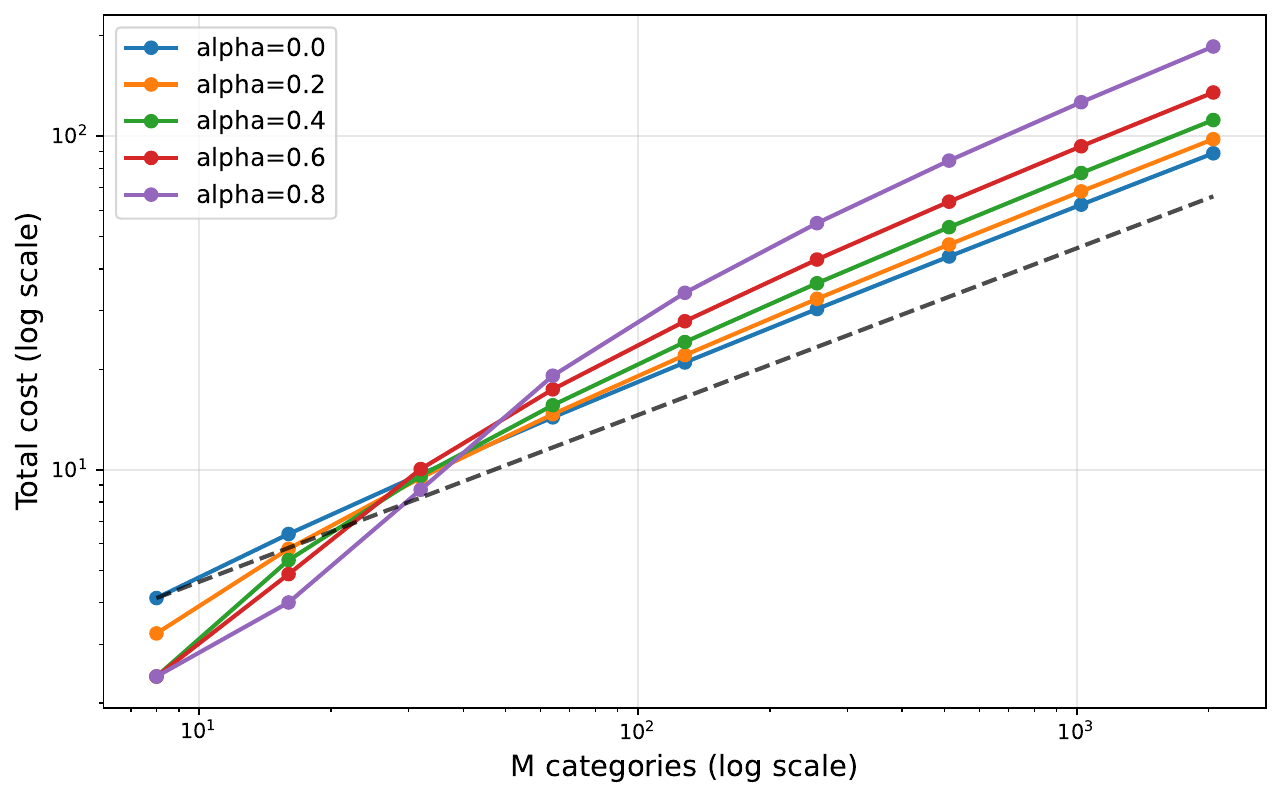}} \\
  \subfloat[Beta$(0.5,1)$]
  {\includegraphics[width=.25\linewidth]{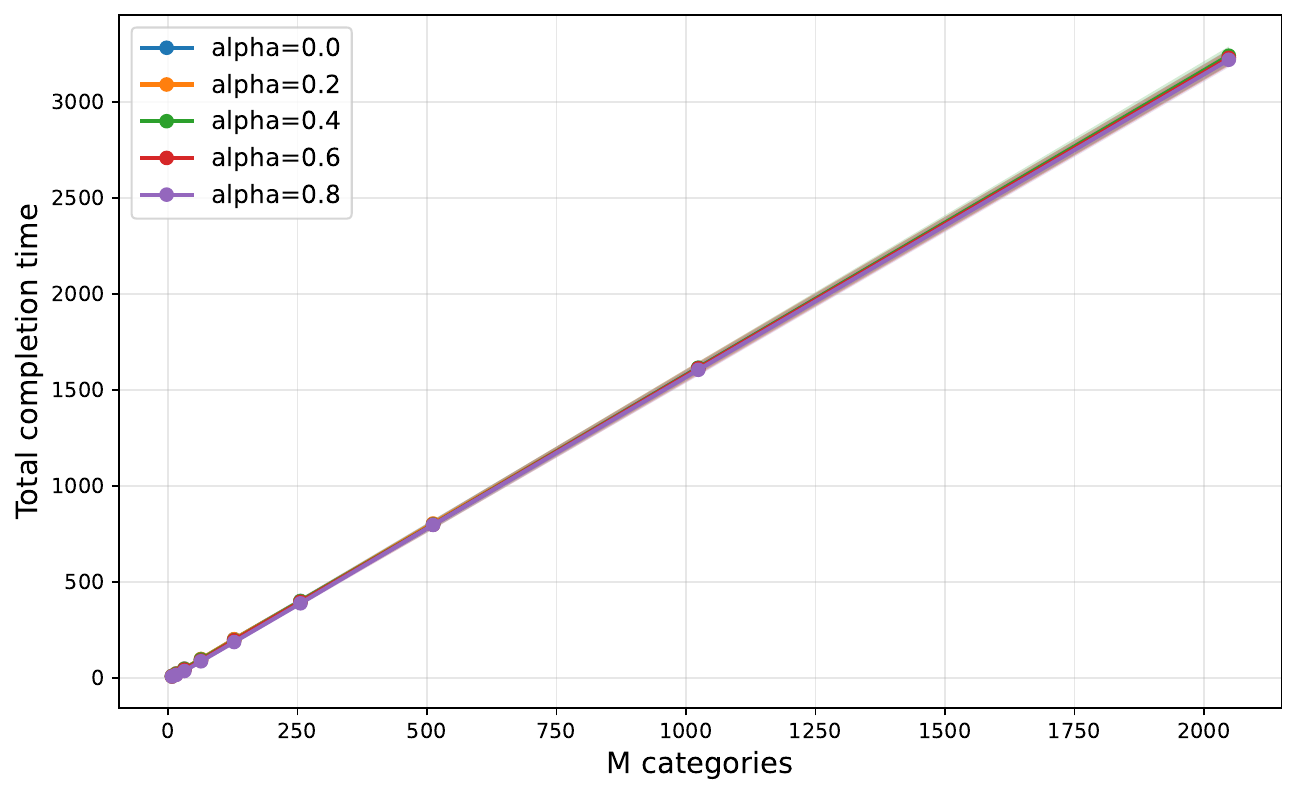}}
  \subfloat[Uniform]
  {\includegraphics[width=.25\linewidth]{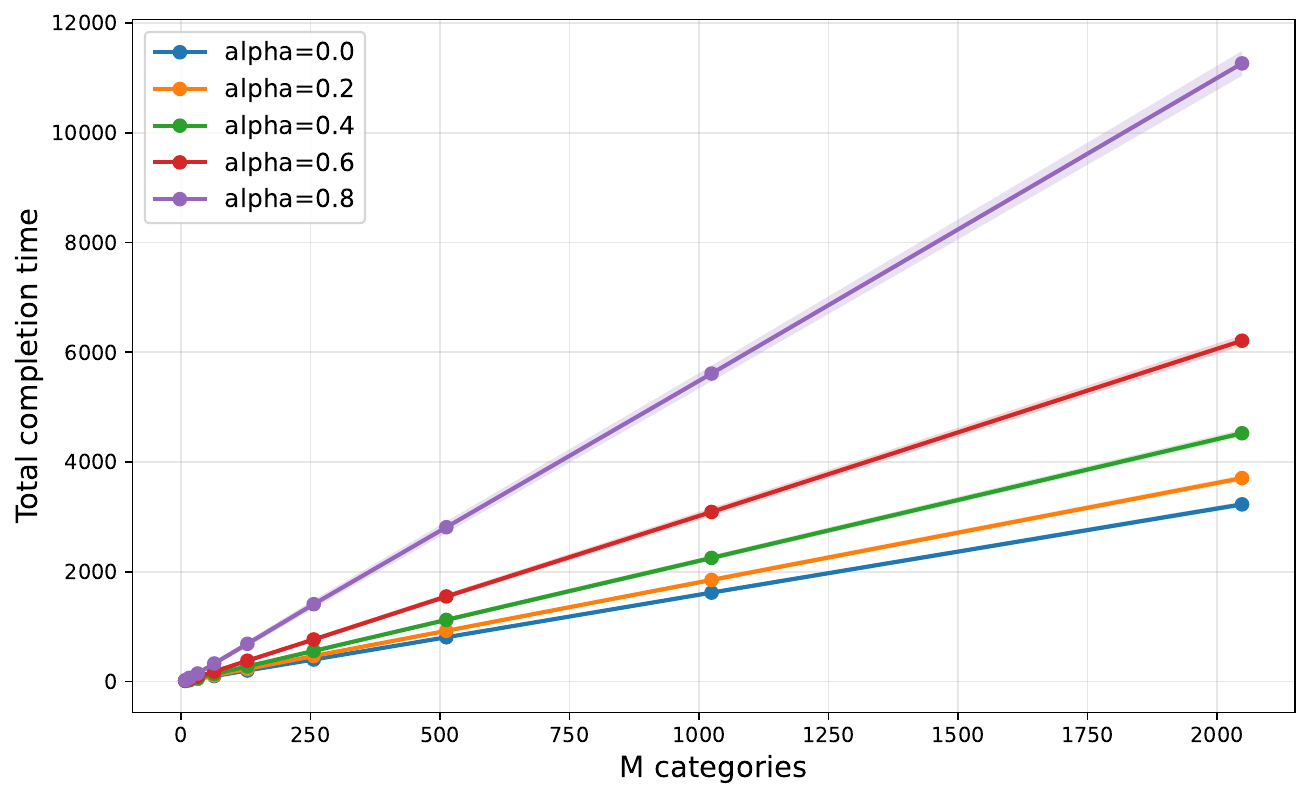}}
  \subfloat[Exponential, $\lambda = 3$]
  {\includegraphics[width=.25\linewidth]{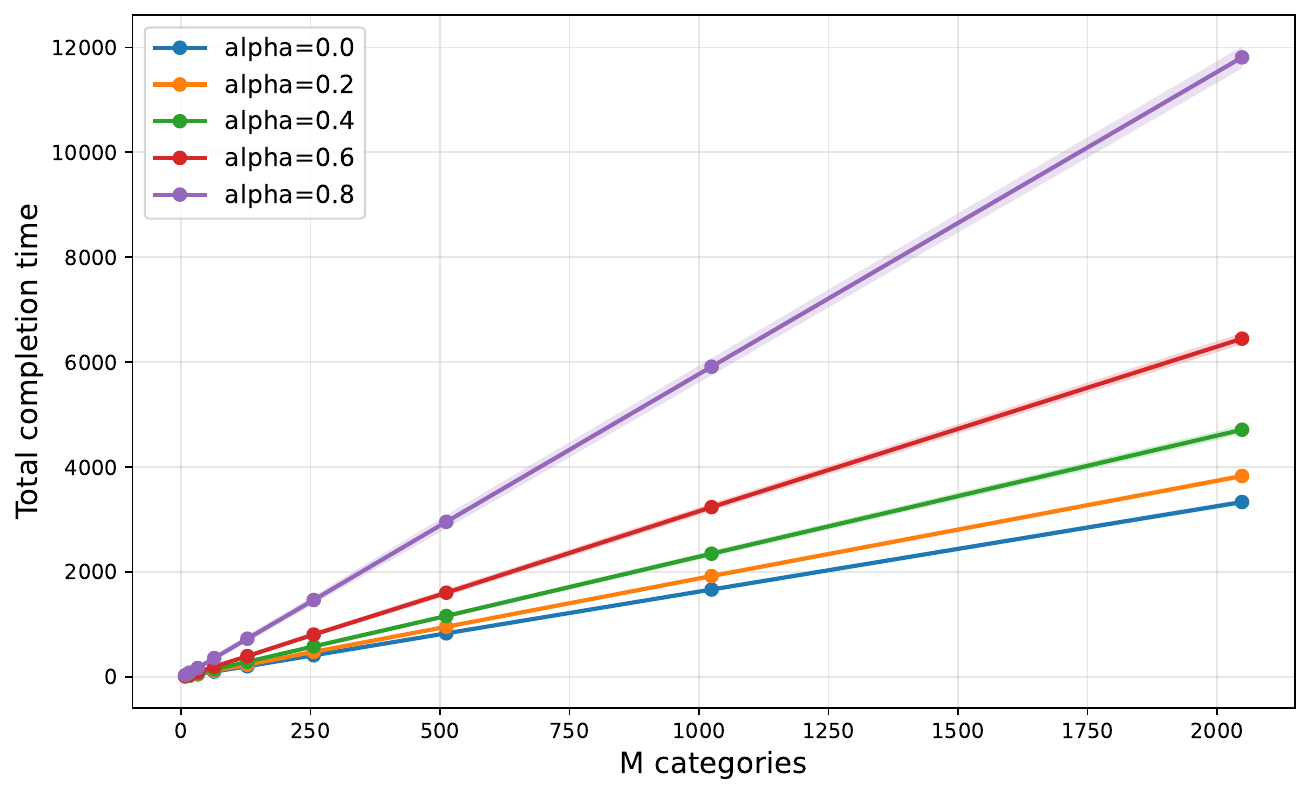}} 
  \subfloat[Beta$(2,1)$]
  {\includegraphics[width=.25\linewidth]{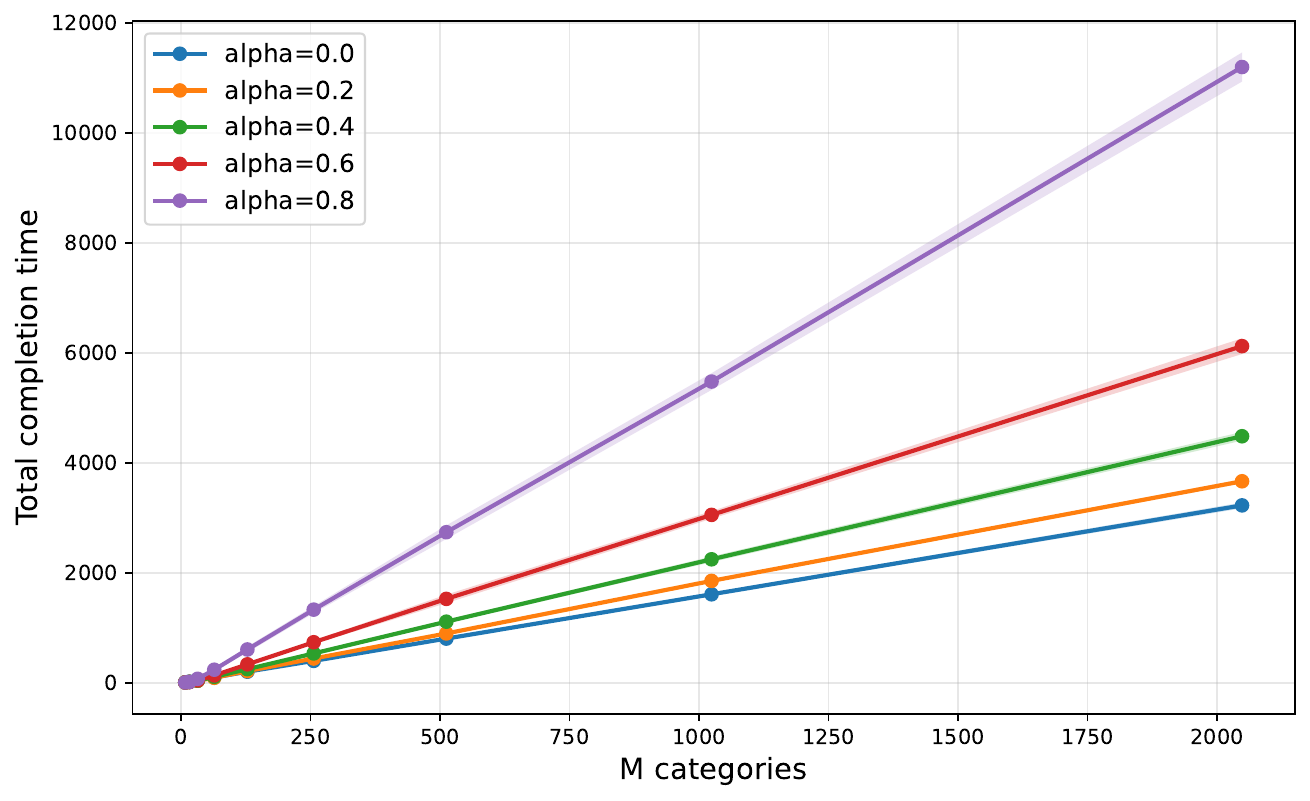}} 
  \caption{Total cost (top row) and total wait time (bottom row) under $\mathbf{p}^{\rm SWS}$ for various values of $\alpha$ and $M$, for a collective price floor value of $0.2$ and different workers' valuations distributions. Panel (a) highlights the $O(1)$ total cost regime, panels (b)--(c) highlight the $O(\log(M))$ total cost regime in log scale, and panel (d) highlights a $O(M^{1/2})$ regime of total cost in log-log scale (dashed line represents the $0.5$ slope line). The total wait time is $O(M)$ with constants depending on $\alpha$ and distribution families.}
  \label{fig:totalcompletiontime_all_quantile}
\end{figure*}

%% file: sec5-targetedcoordination.tex
\section{Vertical collective action as (re)distribution}
\label{sec:targetedcollact}

In contrast to horizontal collectives, \textbf{targeted} collectives strategically allocate limited organizing resources (e.g., pointed information campaigns or price-floor agreements) to the lowest quantiles, where wage suppression originates. In our model, this means focusing on $\alpha$-fraction segments of the market and concentrating a small ``budget'' $\epsilon$ of organizing effort in each by targeting the lowest-cost workers. 

We refer to such targeted strategies as \emph{vertical coordination}, where ``vertical'' means organizing within entire market segments (here, segments correspond to task categories). We draw on classical organizing strategies in collective action: vertical collectives can be more effective in dispersed markets~\citep{olson2012logic}, can address fairness concerns in economic models~\citep{kopel2019endogenous}, and have been applied as industry-region organizing in countries like Germany~\citep{jager2022german}. Both recent~\citep{diani2025embeddedness} and longitudinal~\citep{martin2004conceptualizing,milkman2001organizing} studies show that the shift from trade-based organizational unionizing to more heterogeneous movements across sectors can reduce bargaining power.

Clearly, vertical coordination places workers in the \emph{coordinated markets} regime: for an $\alpha$-fraction of categories, the intervention instills a price floor, and thus the principal must pay a total cost at least $\Theta(\alpha M)$ (as property~\ref{eq:left-gap} is satisfied for an $\alpha$-fraction of categories). We consider the fraction $\alpha$ a constant and show that even for small $\alpha$, Lemma~\ref{lemma:gap-linear} applies and we recover the tight linear cost regime. We next quantify the minimum budget required to deploy such targeted interventions in terms of total variation distance on valuation distributions.~\looseness=-1

\begin{theorem}
    Define the \emph{critical budget} for vertical collective action as 
    \begin{equation}
        \epsilon^* := \inf \{ \epsilon >0: \inf_i Q_i(\epsilon) > 0\}. 
    \end{equation}
    For any $\epsilon > \epsilon^*$, we can construct a new set of valuations distributions $\{D_i^*\}_{i \in [M]}$ such that $\mathrm{TV}(D_i^*,D_i)\le \epsilon$ for all $i \in [M]$, and under which the stochastic wage suppression pricing strategy must incur total cost \[\mathrm{COST}(\mathbf{p}^{\rm SWS}) = \Theta(M).\]
    \label{thm:targetedepsilon}
\end{theorem}

\begin{proof}[Proof sketch]

For any $\epsilon \geq \epsilon^*$, define the valuation distributions after vertical collective action for category $i$ by removing $\epsilon$ mass from the lowest-cost region
and placing it at $1$:
\begin{equation}
    D_i^*(x)\;:=\;\max\{D_i(x)-\epsilon,\,0\}\ \ \text{for }x<1,\qquad D_i^*(1)=1.
\end{equation}

This transformation moves at most $\epsilon$ probability mass, hence $\mathrm{TV}(D_i^*,D_i)\le \epsilon$. We then claim that for every $u\in (0,1-\epsilon)$, the quantiles of the adjusted distribution satisfy $Q_i^*(u)\;\ge\; Q_i(u+\epsilon)$ (proved in Proposition~\ref{prop:quantileineq} in the Appendix). 
Next, we note that the $\mathbf{p}^{\rm SWS}$ strategy satisfies a constant probability of success of at least $1 - e^{-\theta}$ for some $\theta> 0$ in each iteration $t$ (more precisely, for $\theta > 1 - \epsilon)$, since it is defined as 
\begin{equation}
    p_t \;\ge\; \max_{i\in A_t} Q_i^*\!\Big(\frac{\theta}{M-t+1}\Big)
\;\ge\;
\max_{i\in A_t} Q_i\!\Big(\epsilon+\frac{\theta}{M-t+1}\Big)
\;\ge\;
\max_{i\in A_t} Q_i(\epsilon).
\end{equation}

Hence the total payment obeys
\begin{equation}
\begin{gathered}
    \mathrm{COST}(\mathbf{p}^{\rm SWS}) \;=\; \sum_{t=1}^M p_t
\;\ge\; \sum_{t=1}^M \max_{i\in A_t} Q_i(\epsilon)
\;\ge\; M \cdot \underline{q}_\epsilon, \\
\quad\text{where}\quad
\underline{q}_\epsilon \;:=\; \min_{i\in[M]} Q_i(\epsilon).
\end{gathered}
\end{equation}

In particular, as long as $\underline{q}_\epsilon>0$ (a constant independent of $M$), the principal's cost is $\Theta(M)$, with a wait time of $O(M)$ with constants depending on $\epsilon$. Without this condition, the principal can still exploit tail behavior and achieve sublinear total cost: we need that the lowest $\epsilon$-mass of the distribution is not ``free'', so the workers in this region do not take arbitrarily low payments. In that case, we need a critical budget high enough to cover the entire mass.
\end{proof}

\begin{remark}
    For clarity, in Theorem~\ref{thm:targetedepsilon} the intervention is applied to all categories; applying it to an $\alpha$-fraction forces cost $\Theta(\alpha M)$. We connect Theorem~\ref{thm:targetedepsilon} with different market regimes in terms of the critical budget necessary to prevent wage suppression:~\looseness=-1

    \begin{itemize}
        \item \textbf{No campaign is needed in coordinated markets:} in markets where there is already a positive price floor (property~\eqref{eq:left-gap}), collective action is not necessary, since the total cost is $\Theta(M)$.

    \item \textbf{When any vertical campaign prevents wage suppression:} if targeted categories have no atom at $0$ (and whose support's left endpoint becomes positive after any $\epsilon$-shift), then $\epsilon^*=0$ for those categories. This can restore any sublinear wage suppression regime to a linear one, including logarithmic and constant. This is significant as it includes a wide range of distributions, as described in Section~\ref{sec:nocoordination}, and hence various market conditions.
    
    \item \textbf{When a vertical campaign needs a minimum budget:} if more than $(1 - \alpha)$-fraction of the valuations distributions have an atom at $0$, i.e. $\epsilon^* = \rho$ for some $\rho > 0$, that means that a vertical collective must `move' the entire atom in order to shift the total price from $O(1)$ to $O(M)$. This implies a choice on the side of the decision-maker: choose a target fraction $\alpha$ that allows targeting distributions without this property, or choose a high enough budget $\epsilon$ to cover the mass at $0$. Our model allows a decision-maker to make this choice by providing a guarantee of \emph{when} targeting works and what budget it requires.
    \end{itemize}    
\end{remark}

This shows the \textbf{power of targeted recruiting} in collective action: horizontal collectives preserve wage suppression, whereas vertical collectives restore the cost bound to $\Theta(M)$ for any small budget as long as an $\alpha$-fraction of categories do not have an atom at $0$. Instead of thinning arrivals, they lift the free-labor region. When such regions have significant mass, collectives are still effective as long as the targeting budget exceeds the atom size. This difference captures the classic tension between breadth and focus in collective action: horizontal collectives diffuse effort across many categories, while vertical collectives spend their limited $\epsilon$ budget where it matters most---on the lowest quantiles that cause wage suppression.~\looseness=-1

Figure~\ref{fig:horizontalverticalcollact} illustrates the asymptotic difference between the power of a collective in horizontal versus vertical coordination: for any fixed $\alpha < 1$, a principal can obtain the $O(M)$ wait time guarantee with the stochastic wage suppression strategy in horizontal coordination (in red), essentially keeping the total cost of order $O(\log(M))$. However, vertical coordination (in green) quickly forces a wage increase as $M$ increases, even for small collective sizes and a small budget $\epsilon$ for targeting, restoring a total cost of $\Theta(M)$.~\looseness=-1

%% file: sec6-conclusion.tex
\section{Discussion}
\label{sec:conclusion}

Our work has implications for both platform designers and worker advocates. On the one hand, it provides a versatile theoretical model for diagnosing when algorithmic pricing can lead to wage suppression, based on the structure of workers’ perceived cost distributions. On the other hand, it offers concrete guidance for designing collective interventions that counteract procurement power, highlighting when and why targeted organizing can be effective.

Several open directions emerge. First, our pricing strategy could be extended to dynamic pricing models that learn workers' valuation distributions online from task uptake. While the stochastic wage suppression strategy we analyze is not intended to be cost-optimal, it provides a natural benchmark under partial distributional knowledge and reveals sharp phase transitions in workers' wages. Second, future work could evaluate platform-side interventions that reduce workers' uncertainty about their costs and study how these interact with different models of collective action, including alternative targeting and redistribution mechanisms. In this work, we considered $\alpha$ as an \emph{a priori} chosen collective fraction. The problem of dynamically adjusting $\alpha$ with $M$ is also interesting as a future direction. Third, an important empirical direction is to characterize valuation distributions in real-world platforms and to examine how they evolve under sustained pricing strategies and organizing efforts.

Overall, our results underscore a structural asymmetry in digital labor markets: while random, untargeted coordination tends to thin participation without altering wage suppression regimes, even modestly targeted collective action can fundamentally change market outcomes.

\subsection*{Acknowledgements}

Celestine Mendler-Dünner acknowledges financial support of the Hector Foundation.

%% file: sec-appendix.tex
\section{Appendix}
\label{sec:appendix}

We detail proofs of the main results presented in the main text. 

\subsection{Wait time of SWS}

We start with basic properties of the wait time and success probability at iteration $t$: since workers are drawn sequentially until the first success occurs, $N_t \sim \mathrm{Geom}(q_t)$, with mean and variance
\begin{equation}
    \mathbb{E}[N_t] = \frac{1}{q_t}, 
\qquad \mathrm{Var}(N_t) = \frac{1-q_t}{q_t^2}.
\end{equation}

As defined before, the total number of workers needed across all iterations is equal to the total wait time ${\rm WAIT}(\mathbf{p}^{\rm SWS}) := \sum_{t=1}^M N_t$. Clearly, $N_t \perp N_{t + 1}$, so
\begin{equation}
    \mathbb{E}[{\rm WAIT}(\mathbf{p}^{\rm SWS})] = \sum_{t=1}^M \frac{1}{q_t},
\qquad
\mathrm{Var}({\rm WAIT}(\mathbf{p}^{\rm SWS})) = \sum_{t=1}^M \frac{1-q_t}{q_t^2}.
\end{equation}

As part of our proofs, we also need a standard result on sub-exponential behavior of Geometric variables:  

\begin{lemma}[Sub-exponential property of Geometric random variables]
\label{lem:geom-subexp-stage}
If $\Delta\sim\mathrm{Geom}(q)$ on $\{1,2,\dots\}$, then
$\mathbb{E}\Delta = 1/q$, $\mathrm{Var}(\Delta)=(1-q)/q^2$.
Moreover, there exists a universal $C>0$ such that
\begin{equation}
\|\Delta-\mathbb{E}\Delta\|_{\psi_1} \;\le\; \frac{C}{q},
\end{equation}
i.e., $\Delta-\mathbb{E}\Delta$ is sub-exponential with scale $O(1/q)$.
\end{lemma}

\begin{proof}
The mgf of $\Delta$ is $M_\Delta(\lambda)=\frac{q e^\lambda}{1-(1-q)e^\lambda}$ for $\lambda<-\ln(1-q)$. For $|\lambda|\le c\,q$ with a small enough absolute $c>0$, a Taylor expansion of $\log M_\Delta(\lambda)-\lambda\,\mathbb{E}\Delta$ yields
$\log \mathbb{E} e^{\lambda(\Delta-\mathbb{E}\Delta)}\le C_1\lambda^2/q^2$. Since this holds for $|\lambda| \leq cq$, we also get $\log \mathbb{E} e^{\lambda(\Delta-\mathbb{E}\Delta)}\le C_1 c \lambda/q$, which is the standard criterion for sub-exponentiality with parameter $O(1/q)$.
\end{proof}

\begin{proof}[Proof of Theorem~\ref{thm:quantimeconcentration}]

By Lemma~\ref{lemma:quantileprices}, we have $q_t \ge q_*:=1-e^{-\theta}$ for all $t\in[M]$.
Since $N_t\sim \mathrm{Geom}(q_t)$ and the $N_t$ are independent across $t$,
\[
\mathbb{E}[{\rm WAIT}(\mathbf{p}^{\rm SWS})]
=\sum_{t=1}^M \mathbb{E}[N_t]
=\sum_{t=1}^M \frac{1}{q_t}
\le \frac{M}{q_*}.
\]
We also note that there exists a constant $0<q_-<1$ (depending only on $\theta$) such that
\begin{equation}\label{eq:q-bounds}
q_- \;\le\; q_t \qquad \text{for all } t\in\{1,\dots,M\}.
\end{equation}
as per Lemma~\ref{lemma:quantileprices}, with $q^*$ as defined above. As such, we can also compute the variance of the wait time under the pricing strategy $\mathbf{p}^{\mathrm{SWS}}$ as
\begin{equation}
\begin{gathered}
    \mathrm{Var}({\rm WAIT}(\mathbf{p}^{\rm SWS}))
\;\le\; \sum_{t=1}^{M}\frac{1-q_-}{q_-^2} \;=\; O(M).
\end{gathered}
\end{equation}

For the concentration bound, apply Lemma~\ref{lem:geom-subexp-stage} with $q=q_t\ge q_{*}$ to find $\psi_1$ control $\|N_t-\mathbb{E} N_t\|_{\psi_1}\le C/q_{*}$ for some constant $C$. Since all $N_t$ are independent, we can apply Bernstein’s inequality for the sum of independent sub-exponential variables $X_t = N_t - \mathbb{E}[N_t]$, which sum up to ${\rm WAIT}(\mathbf{p}^{\rm SWS})$. 

We thus find that exists a constant $c=c(\theta)>0$ s.t. $\forall x>0$,
\begin{equation}
\mathbb{P} \!\big({\rm WAIT}(\mathbf{p}^{\rm SWS})-\mathbb{E}[{\rm WAIT}(\mathbf{p}^{\rm SWS})]\ge x\big)
\;\le\;
\exp\!\Big(-\,c\,\min\{\,x^2/M,\; x\,\}\Big),
\end{equation}
or, equivalently, $\exists c^* > 0$ that depends only on $c$ s.t. $\forall \epsilon\in(0,1)$,
\begin{equation}\label{eq:bernstein-relative-stage}
\Pr\!\left(\left|\frac{{\rm WAIT}(\mathbf{p}^{\rm SWS})}{\mathbb{E}[{\rm WAIT}(\mathbf{p}^{\rm SWS})]}-1\right|\ge \epsilon\right)
\;\le\; 2\exp\!\big(-c^*\,\epsilon^2 M\big).
\end{equation}

\end{proof}

We note that the wait time concentration bound is distribution-agnostic and holds for any ${D_i}$ under the independence assumptions of the model.~\looseness=-1

\begin{table}[h]
  \caption{Left-tail properties of $D$ and implications for the total cost achieved by $\mathbf{p}^{\rm SWS}$. }
$$
\begin{array}{c|c}
\text{Left-tail class} & \text{Total cost under }\mathbf{p}^{\rm SWS} \\
\hline
\text{Gap } (g>0) & \Theta(M) \\
D(x)\asymp x & \Theta(\log M) \\
D(x)\asymp x^{\alpha},\ \alpha>1 &  \Theta(M^{\,1-1/\alpha}) \\
D(x)\gtrsim x^{\beta},\ \beta<1 & O(1) \\
\text{Atom at }0 & O(1)
\end{array}
$$
\label{table:costregimes}
\end{table}

\subsection{Total cost of SWS}

See Table~\ref{table:costregimes} for different left-tail properties of $D$.\footnote{In Table~\ref{table:costregimes}, the $\asymp$ notation means that if $f(x) \asymp g(x),$ then $\exists c_1, c_2 > 0$ s.t. $ c_1 g(x) \leq f(x) \leq c_2 g(x)$ for small $x$; the $\gtrsim$ notation means that if $f(x) \gtrsim g(x),$ then $\exists c>0$ s.t. $ f(x) \geq c g(x)$ for small $x$.}.

\begin{proof}[Proof of Lemma~\ref{lemma:gap-linear}]
We show that the stochastic wage suppression strategy satisfies
\begin{equation}
M\,g_{*} \;\le\; \mathrm{COST}(\mathbf{p}^{\rm SWS}) \;\le\; M\cdot \max_{i\in[M]} Q_i(\theta)
\;\le\; M,
\end{equation}
and thus $\mathrm{COST}(p)=\Theta(M)$ (with constants depending only on $g_{*}$ and $\theta$).
If $D_i(x)=0$ for all $x<g_i$, then $Q_i(u)\ge g_i$ for every $u\in(0,1]$. Thus for every iteration $t$, $p_t=\max_{i\in A_t} Q_i(\theta/(M-t+1))\ge \max_{i\in A_t} g_i \ge g_{*},$
so $\sum_t p_t\ge M g_{*}$. For the upper bound, use the monotonicity of $Q_i$ in its argument:
$Q_i(\theta/(M-t+1))\le Q_i(\theta)\le 1$, hence $p_t\le \max_i Q_i(\theta)\le 1$ and the sum is at most $M$.
\end{proof}

We note that to achieve $\Theta(M)$, only a fraction of categories need a gap: if $\alpha M$ categories (for some constant $\alpha\in(0,1]$) satisfy $g_i\ge g_0>0$, then $\mathrm{COST}(\mathbf{p}^{\rm SWS})\ge \alpha g_0\,M$ and still $\mathrm{COST}(\mathbf{p}^{\rm SWS})\le M$.~\looseness=-1

\begin{proof}[Proof of Lemma~\ref{lemma:gapless}]
Fix $\varepsilon>0$. Since $\inf\{x: D_i(x)>0\}=0$, we have $Q_i(u)\to 0$ as $u\downarrow 0$ for every $i$.
Hence, for each $i$ there exists $\delta_i>0$ such that $Q_i(u)\le \varepsilon$ for all $u\in(0,\delta_i]$.
Let $\delta:=\min_{i\in[M]}\delta_i>0$. Let $J:=\left\lceil \theta/\delta \right\rceil$. For all $j\ge J$, we have $\theta/j \le \delta$, and therefore
\[
Q_{\max}(\theta/j)=\max_{i}Q_i(\theta/j)\le \varepsilon.
\]
Writing the total cost as $\sum_{j=1}^M Q_{\max}(\theta/j)$ (with $j=M-t+1$),
we obtain
\[
{\rm COST}(\mathbf{p}^{\rm SWS})
\le \sum_{j=1}^{J-1} Q_{\max}(\theta/j) + \sum_{j=J}^{M} \varepsilon
\le C(\varepsilon)+\varepsilon M,
\]
where $C(\varepsilon):=\sum_{j=1}^{J-1} Q_{\max}(\theta/j)$ is finite and independent of $M$.
Dividing by $M$ and taking $M\to\infty$ yields $\limsup_{M\to\infty} {\rm COST}(\mathbf{p}^{\rm SWS})/M\le \varepsilon$.
Since $\varepsilon$ was arbitrary, ${\rm COST}(\mathbf{p}^{\rm SWS})=o(M)$.
\end{proof}

\begin{observation} 
Canonical examples of distributions with property~\ref{prop:gapless}: 
\begin{itemize}
\item \textbf{Beta}$(a,1)$ on $[0,1]$ with $a>1$ (exactly $D(x)=x^a$) gives a total cost of $\mathrm{COST}(\mathbf{p}^{\rm SWS})=\Theta(M^{1-1/a})$. Similarly, the \textbf{Gamma}$(k > 1,\lambda)$ on $[0,\infty)$ with $k>1$ and the \textbf{Weibull}$(k > 1,\lambda)$ distributions reduce to the same computations.
\item \textbf{Weibull-at-zero type} distribution (stretched–exponential left tails). Take a CDF of the form $D(x)=\exp\!\bigl(-\lambda\,x^{-\beta}\,(1+o(1))\bigr)\quad\text{as }x\rightarrow 0,$ for $\beta,\lambda>0$. Then
\begin{equation}
Q(u)=\Bigl(\frac{\lambda}{\ln(1/u)}\Bigr)^{\!1/\beta}(1+o(1)),
\end{equation}
giving $\mathrm{COST}(\mathbf{p}^{\rm SWS})=\Theta\!\Big(\frac{M}{(\ln M)^{1/\beta}}\Big).$ 
This is \emph{sublinear but can get arbitrarily close to linear}: for any fixed $\epsilon>0$ and large $M$,
$\;M/(\ln M)^{1/\beta} \gg M^{1-\epsilon}$.
\item \textbf{Lognormal left tail.} Take the $\operatorname{LogNormal}(\mu,\sigma^2)$ with a super-light tail. For small $u$,
\begin{equation}
    Q(u)=\exp\!\Big(\mu+\sigma\,\Phi^{-1}(u)\Big)
=\exp\!\Big(\mu-\sigma\sqrt{2\ln(1/u)}\,(1+o(1))\Big).
\end{equation}
With $u=\theta/|A_t|$,
\begin{equation}
Q(\theta/|A_t|)=\exp\!\Big(\mu-\sigma\sqrt{2\ln |A_t|}\,(1+o(1))\Big),
\end{equation}
Taking the integral gives a total cost bound of 
\begin{equation}
    \mathrm{COST}(\mathbf{p}^{\rm SWS})=\sum_{j=1}^M Q(\theta/j)
=\Theta\!\Big(M\,\exp\!\big(-\sigma\sqrt{2\ln M}\big)\Big).
\end{equation}
\end{itemize}
\label{obs:sublinear}
\end{observation}

\begin{proof}[Proof of Lemma~\ref{lemma:logm}]
Property~\ref{condition:left} implies that there are constants $c,t_0 > 0$ such that for all $t > t_0$ and all category indices $i$, $\mathbb{P}[v_i \leq 1/t] \geq c/t$. Thus, for $x \leq 1 / t_0$,  $D_i(x)\geq cx$, or equivalently, $Q_i(u) \leq u/c$. Thus, we can split the total cost in two sets: a set of $j \geq \lfloor \theta t_0 / c\rfloor$ for which $p_t \leq \frac{\theta}{cj}$ and a set $j \leq \lfloor \theta t_0 / c\rfloor$ which is finite and constant (denoted by $C$) with respect to $M$. Then, the total cost can be bounded as 

    \begin{equation}
        \mathrm{COST}(\mathbf{p}^{\rm SWS}) \;=\; \sum_{t=1}^M p_t \leq C + \sum\limits_{j = \lfloor \theta t_0 / c\rfloor}^M \frac{\theta}{cj} = O(\log(M))
    \end{equation}

\end{proof}

\begin{observation} 
Canonical examples of distributions with property~\ref{condition:left} (note that any log-concave law with pdf $f(0)>0$ will satisfy this bound): 
\begin{itemize}
\item \textbf{Uniform $[0,1]$.} $Q(u)=u$, so $p_t=\theta/(M-t+1)$ and ${\rm COST}=\theta\sum_{j=1}^M \frac{1}{j}= \Theta(\log M)$.
\item \textbf{Exponential $(\lambda)$ truncated to $[0,1]$.} Near $0$, $D(x)\sim \lambda x$, hence property~\ref{condition:left} holds with $c=\lambda/2$ for small enough $x_0$, so ${\rm COST}(\mathbf{p}^{\rm SWS})=O(\log M)$. A direct calculation of the quantiles of $D(x) = 1 - e^{-\lambda x}$ gives $Q(1/t) = -1/\lambda \ln(1 - 1/t) = \Theta(1/t)$ through a simple Taylor expansion, thus summing to ${\rm COST}(\mathbf{p}^{\rm SWS})=\Theta(\log M)$.
\item \textbf{Beta}$(1,b)$ with $b < 1$: we can compute the quantiles as $Q(1/t) = 1 - (1 - 1/t)^{1/b} = \Theta(1/t)$ by a similar Taylor expansion, obtaining again ${\rm COST}(\mathbf{p}^{\rm SWS})=\Theta(\log M)$.
\end{itemize}
\label{obs:disperseddistributions}
\end{observation}

We formulate a general criterion for the total price to stay constant with respect to $M$: 
 
\begin{proposition}
\label{prop:O1-iff}
The total cost of the principal is $O(1)$ with respect to $M$ if there exist constants $c_1,c_2>0$ such that for all $M\ge 2$,
\begin{equation}
    \begin{gathered}    
c_1 \int_{1}^{M} Q_{\max} \Bigl(\frac{\theta}{x}\Bigr)\,dx
\;\le\;
\sum_{j=1}^{M} Q_{\max} \Bigl(\frac{\theta}{j}\Bigr) \\ 
\;\le\;
c_2\Biggl(Q_{\max}(\theta)+\int_{1}^{M} Q_{\max}\Bigl(\frac{\theta}{x}\Bigr)\,dx\Biggr),
    \end{gathered}
\end{equation}
for $Q_{\rm max}(u) = \max_i Q_i(u)$. Equivalently, with the change of variables $u=\theta/x$,
\begin{equation}
    \int_{1}^{M} Q_{\max}\!\Bigl(\frac{\theta}{x}\Bigr)\,dx
= \theta\!\int_{\theta/M}^{\theta}\frac{Q_{\max}(u)}{u^{2}}\,du.
\end{equation}

The total cost is finite if and only if $\int_{0^+}^{\theta} \frac{Q_{\max}(u)}{u^{2}}\,du < \infty.$
\end{proposition}

\begin{proof}[Proof]
The proposition follows from the nondecreasing property of  $Q_{\max}(u)$ and applying the integral test for monotone sequences to
$Q_{\max}(\theta/x)$ with the mentioned change of variables.
\end{proof}

\begin{proof}[Proof of Lemma~\ref{lemma:constant}]
If there exist $x_0\in(0,1]$, $\beta\in(0,1)$ and $c>0$ such that for all $i$ and $x\in[0,x_0]$, $D_i(x)\;\ge\; c\,x^{\beta}$, we can compute the quantiles for all small $u$,
$Q_i(u)\le (u/c)^{1/\beta}$. Thus,
\begin{equation}
p_t \le Q_{\max}(\theta/j)\;\le\; ( \theta/c)^{1/\beta}\, j^{-1/\beta}.
\end{equation}
Since $1/\beta>1$, $\sum_{j\ge 1} j^{-1/\beta}<\infty$, $\sum_t p_t=O(1)$ w.r.t. $M$. Finally, if the distributions $D_i$ have an atom at $0$, it means that $D_i(0)\ge \rho_i>0$, then for $j\ge \theta/\rho_i$, $Q_i(\theta/j)\le 0$. Given nonnegative prices, only finitely many terms are positive, so $\sum_t p_t=O(1)$  w.r.t. $M$.~\looseness=-1
\end{proof}

\begin{observation}[Phase transition behavior for the ${\rm Beta}(a,1)$ distribution family]
\label{prop:beta}
The Beta distribution can be used to illustrate all regimes of growth for the total price of the principal under the stochastic wage suppression pricing strategy discussed in the main text. If $D(x)=x^a$ (i.e., $X\sim\mathrm{Beta}(a,1)$), then $Q(u)=u^{1/a}$, so $p_t = \left(\theta/(M-t+1)\right)^{1/a}.$ Then
\begin{equation}
\begin{gathered}
{\rm COST}(\mathbf{p}^{\mathrm{SWS}}) \;=\; \sum_{j=1}^M \Bigl(\frac{\theta}{j}\Bigr)^{1/a}
\;=\;
\begin{cases}
\Theta(M^{1-1/a}), & a>1,\\[2pt]
\Theta(\log M), & a=1,\\[2pt]
\Theta(1), & a<1.
\end{cases}
\end{gathered}
\end{equation}
\end{observation}

\begin{proof}
The quantiles of the ${\rm Beta}$ distribution can be computed in closed form as $Q(u)=u^{1/a}$. Then, the standard asymptotics of $\sum_{j=1}^M j^{-\alpha}$ apply: it diverges like $M^{1-\alpha}$ for $\alpha\in(0,1)$, like $\log M$ for $\alpha=1$, and converges for $\alpha>1$.~\looseness=-1 
\end{proof}

\begin{observation}[On optimality under distributional knowledge.]
Given full knowledge of $\{D_i\}_{i\in A_t}$, the minimum price at iteration $t$ necessary to satisfy the iteration-level success floor $q_{*}$ is
\begin{equation}
p_t^{*} \;\in\; \arg\min_{p\in[0,1]}\Big\{p:\; 1-\prod_{i\in A_t}\bigl(1-D_i(p)\bigr)\;\ge\; q_{*}\Big\}.
\end{equation}

The pricing strategy based on the $1/|A_t|$ quantiles essentially assumes that the principal has partial knowledge, through an oracle access to the $\{Q_i(u): i\in A_t\}$ at $u = 1/|A_t|$.
\label{obs:optimality}
\end{observation}

\subsection{Random collective action }

\begin{proof}[Proof of Lemma~\ref{lemma:horizontaltime}]
For clarity, we present the calculation for the uniform family; the same argument extends to distributions satisfying left-tail regularity $D_i(x)\geq cx$. We compute at iteration $t$ and for an active category $i$,
\begin{equation}
D_i^{\mathrm{mix}}(p_t) \;\ge\; (1-\alpha)\,p_t,
\end{equation}
hence the probability of success in iteration $t$ can be written as
\begin{equation}
\begin{gathered}
q_t \;=\; 1-\prod_{i\in A_t}\bigl(1-D_i^{\mathrm{mix}}(p_t)\bigr)
\;\ge\; 1-\Bigl(1-(1-\alpha)\,p_t\Bigr)^{|A_t|} \\
\;\ge\; 1-\exp\!\bigl(- (1-\alpha)\,\theta\bigr)
\;=:\; q_{*}(\alpha).
\end{gathered}
\end{equation}
Thus, we obtain again that $\Delta_t\sim\mathrm{Geom}(q_t)$ with $q_t\ge q_{*}(\alpha)$, and thus 
\begin{equation}
\begin{gathered}
\mathbb{E}[{\rm WAIT}(\mathbf{p}_{\alpha}^{\rm SWS})]\;\le\;\frac{M}{q_{*}(\alpha)} \;=\; O(M).
\end{gathered}
\end{equation}

The concentration proof around $\mathbb{E}[{\rm WAIT}(\mathbf{p}_{\alpha}^{\rm SWS})]$ follows exactly like in the proof of Theorem~\ref{thm:quantimeconcentration}: for any $x$, $D_i(x)\ge (1-\alpha)x$ because the uniform component contributes $(1-\alpha)x$ and the point mass only increases the CDF. We plug $x=p_t$ and note that the success probabilities are still geometric due to the independence of the $N_t$ random variables. The rest of the proof is identical to Theorem~\ref{thm:quantimeconcentration}, as the tail bound follows from applying Bernstein's inequality for the sum of independent sub-exponential variables, with $q_{*}$ replaced by $q_{*}(\alpha)$.~\looseness=-1
\end{proof}

\begin{proof}[Proof of Lemma~\ref{lemma:randomcollectivecost}]
Let $p^-:=\overline p_i$ and let $D_i(p^-)$ denote the left limit at $p^-$.
For the mixture CDF $D_i^{\rm mix}(x)=(1-\alpha)D_i(x)+\alpha\mathbf{1}\{x\ge p^-\}$, the left-quantile function
$Q_i^{\alpha}$ satisfies
\[
Q_i^{\alpha}(u)=
\begin{cases}
Q_i\!\left(\frac{u}{1-\alpha}\right), & 0 \le u < (1-\alpha)D_i(p^-),\\[2pt]
p^-, & (1-\alpha)D_i(p^-) \le u \le \alpha + (1-\alpha)D_i(p^-),\\[2pt]
Q_i\!\left(\frac{u-\alpha}{1-\alpha}\right), & \alpha + (1-\alpha)D_i(p^-) < u \le 1.
\end{cases}
\]

We partition the sum of the quantiles in three parts over the quantiles $\theta / |A_t|$, a first threshold $u_1  =\lceil \theta / (1 - \alpha)D_i(p^-) \rceil$ and a second threshold $u_2 = \lfloor \theta / (\alpha + (1 - \alpha) D_i(p^-))\rfloor$ for the index set $1,2,\dots, M$. These thresholds are independent of $M$, and thus we have a constant number of them in the price floor regime $p^-$ and in the upper regime, whereas the lower quantile regime is just shifted by a factor of $1/(1-\alpha)$ from the sum in the original distribution $D_i$. We apply the same for taking the maximum over the quantiles in the computation of $\mathbf{p}_{\alpha}^{\rm SWS}$ evaluated at each quantile. 
As such, for the distribution regimes we have analyzed in closed-form as a function $M$, there exists constants $C_1,C_2 > 0 $ s.t.~\looseness=-1 

\begin{equation}
    \frac{1}{1 - \alpha} {\rm COST}(\mathbf{p}^{\rm SWS}) - C_1 \leq {\rm COST}(\mathbf{p}_{\alpha}^{\rm SWS}) \leq \frac{1}{1 - \alpha} {\rm COST}(\mathbf{p}^{\rm SWS}) + C_2
\end{equation}

This is because for any distributions with properties~\ref{prop:gapless},~\ref{condition:left}, or~\ref{condition:poly}, we can use a change of variable $u \rightarrow u/(1 - \alpha)$. In the case where all distributions $D_i$ are the same with quantiles $Q_i := Q$, it is trivial to notice that ${\rm COST}(\mathbf{p}_{\alpha}^{\rm SWS}) = \sum_{j = 1}^{M} Q^{\alpha}\left(\frac{\theta (1-\alpha)}{|A_t|} \right)$ is approximately $\frac{1}{1- \alpha} \int_{1 - \alpha}^{M(1 - \alpha)} Q(\theta / x) dx$. The same argument follows for different distributions $D_i$ that come from the same category (or families with the same tail properties). For distributions that have a positive gap at $0$, collective action is not actually needed, since the total cost paid by the principal is $\Theta(M)$ in the absence of any collective.
\end{proof}

\subsection{Targeted collective action }

\begin{proposition}
    For every $u\in (0,1-\epsilon)$, the quantiles of the adjusted distribution 
    \begin{equation}
    D_i^*(x)\;:=\;\max\{D_i(x)-\epsilon,\,0\}\ \ \text{for }x<1,\qquad D_i^*(1)=1.
\end{equation}
satisfy $Q_i^*(u)\;\ge\; Q_i(u+\epsilon)$ for any $i \in [M]$. 
\label{prop:quantileineq}
\end{proposition}

\begin{proof}
    For an $x$ that satisfies $D_i^*(x) \geq u > 0$, it means that $\max (D_i(x) -\epsilon,0) = D_i(x) - \epsilon$, so $D_i^*(x) = D_i(x) - \epsilon \geq u,$ so $ D_i(x) \geq \epsilon + u$. Coming back to the definition of the quantiles and taking infimums over $x$, we obtain the desired inequality. The inequality does not hold for $u =0$, but that case is not needed in our set-up since our quantile prices are set for non-negative quantiles. Note that for $u \geq 1 - \epsilon$, both quantiles are $1$: $Q_i^*(u) = Q_i(u + \epsilon) = Q_i(1)$, since the distributions have support on $[0,1]$. This holds for any $i \in [M]$. 

Then, we can show that $\mathrm{TV}(D_i^*,D_i)\le \epsilon$ for all $i$. To see this, take an arbitrary category $i\in[M]$ and denote by $\mu_i$ the measure on $[0,1]$ associated to the CDF $D_i$. We assume that $\epsilon < \mu([0,1])$ (so we shave off less than the available total mass). By construction, we had removed the lowest $\epsilon$ mass, but this holds for any measure with mass $\epsilon$, so let's denote by $l$ a measure on $[0,1]$ with $l([0,1]) = \epsilon$. Then, we can define another measure $\nu$ as 

\begin{equation}
    \nu = \mu - l + \epsilon \delta_1,
\end{equation}
where $\delta_1$ is just the Dirac function at $1$ (where we placed the $\epsilon$ mass). Then, we associated $D_i^*$ with $\nu$ and see that ${\rm TV}(D_i^*, D_i) = \epsilon$. By the total variation distance, 

\begin{equation}
    {\rm TV}(D_i^*, D_i) = \sup_{A \subseteq [0,1] } |\mu(A) - \nu(A)|
\end{equation}

For any Borel set $A$, $\mu(A) - \nu(A) = l(A) - \epsilon \delta_{1 \in A}$. If $1 \notin A$, this is equal to $l(A) \leq l([0,1]) = \epsilon$ since $A \subseteq [0,1]$. If $1 \in A$, then its absolute value is less than $\max (l(A) - \epsilon, \epsilon - l(A)) \leq \epsilon$ as above. Finally, take $A = [0,1)$ to see that $\mu(A) - \nu(A) =  l([0,1)) = \epsilon$, hence ${\rm TV}(D_i^*, D_i) = \epsilon$. 
\end{proof}